\documentclass{amsart}
\usepackage{mathrsfs}
\usepackage{amssymb}
\usepackage{amsbsy}
\usepackage{suffix}
\usepackage{mathtools}
\usepackage{booktabs}

\usepackage{color}

\newtheorem{theorem}{Theorem}

\newtheorem{proposition}{Proposition}[section]
\newtheorem{lemma}{Lemma}[section]

\theoremstyle{definition}
\newtheorem{definition}{Definition}[section]

\theoremstyle{remark}

\numberwithin{equation}{section}

\def\XXint#1#2#3{{\setbox0=\hbox{$#1{#2#3}{\int}$}
\vcenter{\hbox{$#2#3$}}\kern-.5\wd0}}

\def\e{\text{e}}
\def\i{\text{i}}

\DeclarePairedDelimiterX\MeijerM[3]{\lparen}{\rparen}%
{\begin{smallmatrix}#1 \\ #2\end{smallmatrix}\delimsixe\vert\,#3}

\newcommand\MeijerG[8][]{%
  G^{\,#2,#3}_{#4,#5}\MeijerM[#1]{#6}{#7}{#8}}

\WithSuffix\newcommand\MeijerG*[7]{%
  G^{\,#1,#2}_{#3,#4}\MeijerM*{#5}{#6}{#7}}
\def\braket#1{\mathinner{\langle{#1}\rangle}}

\let\protect\relax
{\catcode`\|=\active
  \xdef\Braket{\protect\expandafter\noexpand\csname Braket \endcsname}
  \expandafter\gdef\csname Braket \endcsname#1{\begingroup
     \ifx\SavedDoubleVert\relax
       \let\SavedDoubleVert\|\let\|\BraDoubleVert
     \fi
     \mathcode`\|32768\let|\BraVert
     \left\langle{#1}\right\rangle\endgroup}
}
\def\BraVert{\@ifnextchar|{\|\@gobble}
     {\egroup\,\mid@vertical\,\bgroup}}
\def\BraDoubleVert{\egroup\,\mid@dblvertical\,\bgroup}
\let\SavedDoubleVert\relax

{\catcode`\|=\active
  \xdef\set{\protect\expandafter\noexpand\csname set \endcsname}
  \expandafter\gdef\csname set \endcsname#1{\mathinner
        {\lbrace\,{\mathcode`\|32768\let|\midvert #1}\,\rbrace}}
  \xdef\Set{\protect\expandafter\noexpand\csname Set \endcsname}
  \expandafter\gdef\csname Set \endcsname#1{\left\{%
     \ifx\SavedDoubleVert\relax \let\SavedDoubleVert\|\fi
     \:{\let\|\SetDoubleVert
     \mathcode`\|32768\let|\SetVert
     #1}\:\right\}}
}
\def\midvert{\egroup\mid\bgroup}
\def\SetVert{\@ifnextchar|{\|\@gobble}
    {\egroup\;\mid@vertical\;\bgroup}}
\def\SetDoubleVert{\egroup\;\mid@dblvertical\;\bgroup}

\begingroup
 \edef\@tempa{\meaning\middle}
 \edef\@tempb{\string\middle}
\expandafter \endgroup \ifx\@tempa\@tempb
 \def\mid@vertical{\middle|}
 \def\mid@dblvertical{\middle\SavedDoubleVert}
\else
 \def\mid@vertical{\mskip1mu\vrule\mskip1mu}
 \def\mid@dblvertical{\mskip1mu\vrule\mskip2.5mu\vrule\mskip1mu}
\fi

\begin{document}

\title[Representation theory and products of random matrices]{Representation theory and products of random matrices in $\text{SL}(2,{\mathbb R})$}

\author{Alain Comtet}
\address{Laboratoire de Physique Th\'eorique et Mod\`eles Statistiques, B\^{a}timent Pascal 530, rue Andr\'{e} Rivi\`{e}re, Universit\'{e} Paris Sud,
91405 Orsay CEDEX, FRANCE}
              \email{alain.comtet@u-psud.fr}        

\author{Christophe Texier}
\address{Laboratoire de Physique Th\'eorique et Mod\`eles Statistiques, B\^{a}timent Pascal 530, rue Andr\'{e} Rivi\`{e}re, Universit\'{e} Paris Sud,
91405 Orsay CEDEX, FRANCE} 
              \email{christophe.texier@u-psud.fr}        

\author{Yves Tourigny}
\address{School of Mathematics\\
        University of Bristol\\
        Bristol BS8 1TW, United Kingdom}
\email{y.tourigny@bristol.ac.uk}

\thanks{This work
has benefited from the financial support of ``Investissements d'Avenir du LabEx
PALM'', (ANR-10-LABX-0039-PALM), under the project ``ProMAFluM''.
We thank the many colleagues with whom we have, over several years, been able to share
and discuss our ideas; in particular Michel Bauer, Aur\'elien Grabsch, Jean--Marc Luck and Mary Maller. Yves Tourigny gratefully acknowledges the support and the warm hospitality received
in the course of frequent visits to the 
Laboratoire de Physique Th\'eorique et Mod\`eles Statistiques.}

\subjclass[2010]{Primary 15B52, Secondary 60B15}

\date{\today}


\begin{abstract}
The statistical behaviour of a product of independent, identically distributed random matrices in 
$\text{SL}(2,{\mathbb R})$ is encoded in the generalised Lyapunov exponent $\Lambda$; this is a 
function
whose value at the complex number $2 \ell$ is the logarithm of the largest eigenvalue of the transfer operator
obtained when one averages, over $g \in \text{SL}(2,{\mathbb R})$, a certain representation $T_\ell (g)$ associated with the product. We 
study some products that arise from models of one-dimensional disordered systems. These models
have the property that the inverse of the transfer operator takes the form of a second-order difference or differential operator.
We show how the ideas expounded by N. Ja. Vilenkin  in his book [{\em Special Functions and the Theory of Group Representations}, American Mathematical Society, 1968.] can be used to
study the generalised Lyapunov exponent. In particular, we derive explicit
formulae for the almost-sure growth and for the variance of the corresponding products.
\end{abstract} 

\maketitle

\section{Introduction}
\label{introductionSection}
In his remarkable paper \cite{Fu}, 
Furstenberg succeeded in generalising the law of large numbers
to products
$$
\Pi_n := g_n \cdots g_2 g_1
$$
of independent, identically-distributed random elements of a matrix group $G$.
In his theory, an important part is played by a compact homogeneous space, now known as the {\em Furstenberg boundary}, on which the group acts, and the problem of determining the growth rate of the product is reduced to that of finding a certain measure on this boundary that is invariant
under the action of elements drawn at random from the group.
The growth rate is called the {\em Lyapunov exponent} of the product, and the problem of computing it--- or the invariant measure--- given a probability measure on
the group, is notoriously difficult. 

By elaborating ideas that go back to the seminal works of Dyson \cite{Dy}, 
Schmidt \cite{Sc}, Frisch \& Lloyd \cite{FL}, Halperin \cite{Ha}, Kotani \cite{Ko} and Nieuwenhuizen \cite{Ni}, we have, over the past few
years, found several instances
of products of random matrices in $\text{SL}(2,{\mathbb R})$ for which the growth rate can be obtained explicitly in terms of special functions such as the hypergeometric, Bessel, confluent hypergeometric functions and so on; see the survey
\cite{CTT2}. There are also a few cases involving larger matrix groups where calculations have been possible; two such cases were found by Newman \cite{New} and Forrester \cite{Fo}, who succeeded in obtaining formulae for {\em all} the Lyapunov exponents of certain products in terms of the gamma and dilogarithmic functions; the key observation, in these
particular cases, is that the invariant measure coincides with the rotation-invariant measure.

Now, since Vilenkin's classic book \cite{Vi}, it is widely appreciated that there is a close relationship between special functions and the theory of group representation.
It is the purpose 
of the present paper to explain how Vilenkin's ideas may be brought to bear on this problem. To the best of our knowledge, this is, in itself, an original contribution
to the literature on products of random matrices; it has the merit
of providing a unified presentation of the ideas originating from the work of Frisch \& Lloyd, with the practical benefit that calculations can be pushed further, using concepts
of such generality that one may also envisage applications to other groups.
There is some overlap
with the approach and the results
that will appear in a forthcoming publication by one of us \cite{Te}, but here, although we shall not assume any prior knowledge of the subject, the emphasis is on the connections with representation theory. 

The focus of our analysis is the {\em generalised Lyapunov exponent}, which will be defined
in the next paragraph. This is an important mathematical object, relevant in several physical contexts: chaotic dynamics 
and multifractal analysis
\cite{CPV}; fluid dynamics \cite{Va}; classical dynamics
of an oscillator driven by parametric noise \cite{ST1,ZP}; polymers in a random landscape \cite{FLRT}. In particular, the generalised
Lyapunov exponent yields information on the fluctuations of products of random matrices; these fluctuations are of interest in relation with Anderson localisation \cite{DLA,DS,RT,ST2} and also in relation with disordered Ising spin chains \cite{CGG}.

In the remainder of this extensive 
introduction, we provide a summary of those facts that we shall need in order to develop our material.
We then describe the content of the paper in broad terms, state our main results, and provide some indication of how these results relate to the work of others--- postponing the technical details to the remaining sections.
We shall deal exclusively with the case where
$$
G := \text{SL}(2,{\mathbb R})
$$
and then comment at the end of the paper on the possibility of extending
our ideas to other semi-simple groups.

\subsection{The generalised Lyapunov exponent and Tutubalin's transfer operator}
\label{tutubalinSubsection}
Inspired by the work of Furstenberg, Tutubalin \cite{Tu} proved a central limit theorem for products
of random matrices; he showed that, as $n \rightarrow \infty$,
$$
\frac{\ln \left | {\mathbf x} \Pi_n \right | - n \gamma}{\sigma \sqrt{n}}
$$
converges in distribution to a normal random variable.  Here
\begin{equation}
\gamma := \lim_{n \rightarrow \infty} \frac{1}{n} \,{\mathbb E} \left ( \ln \left | {\mathbf x} \Pi_n \right | \right ) 
\label{growthRate}
\end{equation}
and
\begin{equation}
\sigma^2 := \lim_{n \rightarrow \infty} \frac{1}{n} \,{\mathbb E} \left (  \ln^2 \left | {\mathbf x} \Pi_n \right | - n^2 \gamma^2 \right )
\label{variance}
\end{equation}
where ${\mathbf x}$ is an arbitrary non-zero vector, $| \cdot |$ is
an arbitrary norm on the vector space, and the expectation is taken with respect to the probability measure on the group. It is a fact that the limits do not depend on the precise choice of ${\mathbf x}$ and $| \cdot |$.

Both $\gamma$ and $\sigma^2$
may be expressed in terms of the generalised Lyapunov exponent
of the product, defined by
\begin{equation}
\notag
\Lambda({2 \ell})  := \lim_{n \rightarrow \infty} \frac{\ln {\mathbb E} \left ( \left | \Pi_n^t \Pi_n \right |^{\ell} \right )}{n}
\end{equation}
where $g^t$ denotes the transpose of $g$. Here, we have chosen to write the argument of $\Lambda$ in this unusual way to indicate its relationship with the index $\ell$ used by Vilenkin \cite{Vi} and Vilenkin \& Klimyk \cite{VK1} to label a certain family of representations of the group $\text{SL}(2,{\mathbb R})$.
The existence of the limit in this definition is a non-trivial matter but, in case of existence, 
it is again a fact that the limit on the right-hand side does not
depend on the choice of norm.
It is also clear from the definition that the generalised Lyapunov exponent is unchanged
if we replace the $g_n$ in the product by
$g_n^t$. Other equivalent definitions are used in the literature: 
\begin{equation}
\Lambda ({2 \ell}) 
= \lim_{n \rightarrow \infty} \frac{\ln {\mathbb E} \left ( \left | \Pi_n {\mathbf x} \right |^{2 \ell} \right )}{n} 
= \lim_{n \rightarrow \infty} \frac{\ln {\mathbb E} \left ( \left | {\mathbf x} \Pi_n \right |^{2 \ell} \right )}{n}\,.
\label{generalisedLyapunovExponent}
\end{equation}
The last of these formulae yields
\begin{equation}
\gamma = \Lambda'(0) \;\;\text{and}\;\;\sigma^2 = \Lambda''(0)\,.
\label{growthRateAndVariance}
\end{equation}

Tutubalin proved his theorem by expressing the generalised Lyapunov exponent as the largest eigenvalue of a certain operator.
To describe his approach, we begin by remarking that
every matrix
$$
g =  \begin{pmatrix} a & b \\
c & d \end{pmatrix}  \in G
$$
acts on ${\mathbb R}_\ast^2 := {\mathbb R}^2 \backslash \{ {\mathbf 0} \}$ from the right according to
$$
{\mathbf x} \cdot g := \begin{pmatrix} x_1 & x_2 \end{pmatrix} \begin{pmatrix} a & b \\
c & d \end{pmatrix} = \begin{pmatrix} a x_1 + c x_2 & b x_1 + d x_2 \end{pmatrix}\,.
$$
To every $g \in G$, we assign a linear map, denoted $T(g)$, defined on the space $V$ of all functions 
$v : \,{\mathbb R}_\ast^2 \rightarrow {\mathbb C}$ by
\begin{equation}
\left [ T(g) v \right ] ({\mathbf x}) := v ( {\mathbf x} \cdot g ) = v \left ( a x_1 + c x_2, b x_1 + d x_2 \right )\,.
\label{rightRepresentation}
\end{equation}
The map $g \mapsto T(g)$ obtained in this way is a {\em representation} of the group
--- namely, it satisfies the following properties: Firstly,
if $e$ denotes the identity matrix then $T(e)$ is the identity operator; secondly,
$$
\forall\, g_1,\,g_2 \in G\,,\;\;T(g_1 g_2) = T({g_1}) \,T({g_2})\,.
$$
The vector space $V$, on which the operators $T(g)$ are defined, is called the {\em representation space}. 

We introduce the function $1_\ell \in V$
defined by
\begin{equation}
1_\ell ({\mathbf x}) = | {\mathbf x} |^{2 \ell}\,.
\label{unitFunction}
\end{equation}
For every
${\mathbf x} \in {\mathbb R}_\ast^2$ and every $g \in G$, 
we may write
\begin{equation}
\notag
\left [ T (g) {1}_\ell \right ] ({\mathbf x}) = {1}_\ell \left ( {\mathbf x} \cdot g \right ) 
= \left | {\mathbf x} g \right |^{2 \ell}\,.
\end{equation}
Hence
$$
\left | {\mathbf x} \Pi_n \right |^{2 \ell}
= \left [ T \left ( {\Pi_n} \right ) {1}_\ell \right ] ({\mathbf x})\,.
$$
By using the fact that the $g_n$ are independent with the same distribution, we deduce
\begin{equation}
{\mathbb E} \left (  \left | {\mathbf x} \Pi_n \right |^{2 \ell} \right ) =  \left [ {\mathscr T}^n {1}_\ell \right ] ({\mathbf x})
\label{representationFormula}
\end{equation}
where
\begin{equation}
{\mathscr T} := {\mathbb E} \left [ T (g) \right ]\,.
\label{averagedRepresentation}
\end{equation}
In the literature, the operator ${\mathscr T}$ is often called the {\em transfer operator} 
associated with the product of random matrices \cite{BQ,Po}.
{Heuristically, if this operator admits an eigenvalue of largest modulus, then the generalised Lyapunov
exponent will be obtained by taking its logarithm.}

In his paper, Tutubalin proved that this heuristic argument is indeed correct if one assumes that the 
matrices are drawn from a probability distribution that has a density with respect to the Haar measure 
on the group $G$. This hypothesis
has since been weakened considerably; see \cite{BQ, BL} and the references therein.
Important and relevant as this question is, we shall not be concerned with stating or proving general theorems that guarantee the existence of an eigenvalue 
of the transfer operator of largest modulus.  Rather, our aim in this paper is to bring out the practical usefulness of representation
theory by considering specific probability measures on the group for which the heuristic argument leads to explicit calculations of $\gamma$ and $\sigma^2$.


\subsection{Reducibility, equivalence and unitarity of representations}
\label{basicsSubsection}
The analysis of the spectral problem
obviously entails the search for vector spaces invariant
under the action of the transfer operator. 
\begin{definition}
We say that a subspace of $V$ is {\em invariant} under a representation $T$ if, for every $g \in G$, $T(g)$ maps the subspace to itself. $T$ is called {\em reducible} if its representation space $V$ contains a proper invariant subspace.  Otherwise, it is called {\em irreducible}.
\label{invarianceDefinition}
\end{definition}
Every subspace invariant for the representation is also invariant for the transfer
operator.
Now, the particular representation defined by Equation (\ref{rightRepresentation})
is reducible, as the following makes clear.

\begin{definition}
We say that $v$ is {\em even} (respectively {\em odd}) if
$$
\forall\, {\mathbf x} \in {\mathbb R}_\ast^2,\;\; v(-{\mathbf x}) = v({\mathbf x})\;\;\text{(respectively $-v({\mathbf x})$)}\,.
$$
For $\ell \in {\mathbb C}$, we say that $v$ is {\em homogeneous of degree $2 \ell$} if
$$
\forall\, \tau > 0\,,\;\forall\, {\mathbf x} \in {\mathbb R}_\ast^2,\;\; v( \tau {\mathbf x}) = \tau^{2 \ell}
v({\mathbf x})\,.
$$
\label{parityHomogeneityDefinition}
\end{definition}
It is readily verified that the subspace consisting of all smooth {\em even} homogeneous functions of degree $2 \ell$--- which we shall denote by $V_{\ell}$--- is invariant for $T$; the restrictions of $T$ 
and ${\mathscr T}$ to $V_{\ell}$ will be denoted
by $T_\ell$ and ${\mathscr T}_\ell$ respectively.
Another invariant subspace of $V$ is that consisting of all smooth {\em odd} homogeneous functions of degree $2 \ell$; we shall not need to consider this invariant subspace because 
the function ${1}_\ell$ introduced earlier and defined by 
Equation (\ref{unitFunction}) happens to be {even}. 

Although the functions in the space $V_\ell$ are defined on the punctured plane, the requirement of homogeneity and evenness implies that they are completely characterised by the values they take on certain curves. These curves may be understood as parametrisations of the Furstenberg boundary, and there is some flexibility in choosing them; this will prove convenient when we
consider specific probability measures on the group. The important point is that, although
we may work with different curves, it is essentially always the same representation that is
involved.

\begin{definition}
Two representations $\widetilde{T}$ and $T$ with respective representation spaces $V$ and $\widetilde{V}$ are called {\em equivalent}
if there exists an invertible map, say $Q :\,V \rightarrow \widetilde{V}$, such that
$$
\widetilde{T} \,Q = Q \,T\,.
$$
We then say that $\widetilde{T}$ and $T$ are {\em realisations}, in the spaces $\widetilde{V}$ and $V$
respectively, of the same representation.
\label{EquivalenceDefinition}
\end{definition}

Finally, in discussing the spectral problem for the transfer operator, we need to consider
also the spectral problem for its {\em adjoint}, and this raises the technical question 
of finding an appropriate inner product for the space $V_\ell$.
\begin{definition}
The representation $T_\ell$ is called {\em unitary} if $V_\ell$ is equipped with an inner product,
say $\braket{\cdot|\cdot}$, such that
$$
\forall\, g \in G,\;\forall\,f,\,v \in V_\ell,\,\;\;\braket{T_\ell (g) f | T_\ell (g) v} = \braket{f | v}\,.
$$
\label{unitaryDefinition}
\end{definition}
We remark however that, even if an inner product can be found that makes $T_\ell$ unitary, this
property is lost upon averaging over the group, and so will play only a very minor
part in what follows.

Starting with Bargmann \cite{Ba}, the family of representations $T_\ell$ has been thoroughly studied \cite{GGV, Vi,VK1}. We summarise in the following theorem those properties that will be relevant
to our purpose:
\begin{theorem}
The representation $T_\ell$ is {\em reducible} if and only if $\ell \in {\mathbb Z}$. In the reducible
case,
it may be decomposed
into three irreducible subrepresentations: one, denoted $T_\ell^0$, in a space of finite dimension, and two others, denoted
$T_\ell^+$ and $T_\ell^-$ and called ``holomorphic'', in spaces of infinite dimension.

For $\ell \in {\mathbb Z}$ and $\varepsilon \in \{ -,0,+\}$ (respectively $\ell \notin {\mathbb Z}$), the representations $T_\ell^{\varepsilon}$ and $T_{-\ell-1}^\varepsilon$ (respectively $T_\ell$ and $T_{-\ell-1}$)
are {\em equivalent}.

Of the finite-dimensional representations, only $T_0^0$ is {\em unitary}, as it corresponds to the trivial representation in a one-dimensional space. The holomorphic subrepresentations are unitary with respect to a certain inner product and belong to the so-called discrete series.
The only other values of $\ell$ for which $T_\ell$ can be made unitary are
$$
\ell = -\frac{1}{2} + \text{\rm i} t\,,\;\;t \in {\mathbb R}\;\; \text{and}\;\;
 \ell \in (-2,-1) \cup (0,1)\,.
$$
They are associated with the principal and the supplementary series respectively.
\label{bargmannTheorem}
\end{theorem}

\subsection{The one-parameter subgroups of $\text{SL}(2,{\mathbb R})$}
\label{generatorSubsection}
The proof of Theorem \ref{bargmannTheorem} uses the so-called ``infinitesimal method'' which tackles
questions of irreducibility, equivalence and unitarity by making systematic use of the
{\em infinitesimal generators} associated with the one-parameter subgroups
of $\text{SL}(2,{\mathbb R})$.
These subgroups
may be classified in terms of three types:
\begin{enumerate}
\item \underline{Elliptic}
$$
k(t) := \begin{pmatrix}
\cos \frac{t}{2} & -\sin \frac{t}{2} \\
\sin \frac{t}{2} & \cos \frac{t}{2}
\end{pmatrix}\,.
$$
\item \underline{Hyperbolic}
$$
a_1(t) := \begin{pmatrix}
\e^{\frac{t}{2}} & 0 \\
0 & \e^{-\frac{t}{2}}
\end{pmatrix}
\;\;\text{and}\;\;a_2(t) := \begin{pmatrix}
\text{ch} \frac{t}{2} & \text{sh} \frac{t}{2} \\
\text{sh} \frac{t}{2} & \text{ch} \frac{t}{2}
\end{pmatrix}\,.
$$
\item \underline{Parabolic}
$$
n_+(t) := \begin{pmatrix}
1 & t \\
0 & 1
\end{pmatrix}\;\;\text{and}\;\;
n_-(t) := \begin{pmatrix}
1 & 0 \\
t & 1
\end{pmatrix}\,.
$$
\end{enumerate}

By definition, the generator associated with the one-parameter subgroup ${\mathbb R} \ni t \mapsto g(t) \in G$ is the operator on $V_\ell$ that maps $v$ to
$$
\frac{d}{d t} \left \{ T_\ell \left [ g(t) \right ] v \right \} \Bigl |_{t=0}\,.
$$
A list of the infinitesimal generators for various realisations of $T_\ell$
may be found in Tables \ref{infinitesimalGeneratorTable} to \ref{infinitesimalGeneratorInMellinSpace}.  
We shall consider probability measures that lead to an
explicit formula for ${\mathscr T}_\ell$ in terms of these infinitesimal generators.
The fact that {\em these generators are affine functions
of the index $\ell$} will be important in what follows.

The subgroup $K$ occupies a special place in the harmonic analysis of $G$ as the largest
compact subgroup, and there is a countable basis of $V_\ell$ consisting of eigenvectors
of the corresponding infinitesimal generator. We shall denote these basis functions 
${\mathbf e}_{\ell,n}$, $n \in {\mathbb Z}$, regardless of the particular realisation of the representation
$T_\ell$. 
We draw attention to the important fact that the function $1_\ell$, defined by Equation (\ref{unitFunction}), is a multiple
of ${\mathbf e}_{\ell,0}$.

In order to work out how each generator acts on the basis functions, it is helpful to introduce a basis of the Lie algebra consisting of the operators
\begin{equation}
{J}_0 := \i {K}\,,\;\; {J}_\pm := {A}_1 \pm \i {A}_2\,.
\label{jOperators}
\end{equation}
For instance, we have
\begin{equation}
{N}_\pm = \i \left [ \frac{{J}_--{J}_+}{2} \pm {J}_0 \right ]\,.
\label{firstNilpotentIdentity}
\end{equation}
The action of  $N_\pm$ is then easily deduced from the identities
\begin{equation}
{J}_0 \,{{\mathbf e}_{\ell,n}} = n \,{{\mathbf e}_{\ell,n}}
\label{eigenvectorsOfJ0}
\end{equation}
and
\begin{equation}
{J}_+ \,{{\mathbf e}_{\ell,n}} = (\ell-n) \,{{\mathbf e}_{\ell,n+1}}\;\;\text{and}\;\;{J}_- \,{{\mathbf e}_{\ell,n}} = (\ell+n) \,{{\mathbf e}_{\ell,n-1}}\,.
\label{raisingLowering}
\end{equation}
In particular, the {\em Casimir operator} associated with the representation is defined by
\begin{equation}
\notag
{C} := A_1^2 + A_2^2 -K^2 = {J}_0^2 + {J}_0 + {J}_-{J}_+\,.
\end{equation}
This operator commutes with the representation, and it is readily verified that it
is a multiple of the identity operator ${I}$:
\begin{equation}
\notag
{C} = \ell (\ell+1) \,{I}\,.
\end{equation}

\begin{table}
\begin{tabular}{c  c  c  c c}
& Plane & Circle & Line & Hyperbola \\
\hline
\hline
 & & & & \\
${K}$ & $\frac{x_2}{2} \frac{\partial}{\partial x_1}-\frac{x_1}{2} \frac{\partial}{\partial x_2}$ & $-\frac{d}{d \theta}$ & $- \ell x + \frac{1+x^2}{2} \frac{d}{dx}$ & $\frac{\ell}{2} (\frac{1}{x}-x) + \frac{1+x^2}{2} \frac{d}{dx}$ \\
& & & & \\
${A}_1$ &  $\frac{x_1}{2} \frac{\partial}{\partial x_1}-\frac{x_2}{2} \frac{\partial}{\partial x_2}$ & $\ell \cos \theta -\sin \theta \frac{d}{d \theta}$ & $-\ell + x \frac{d}{dx}$ & $x \frac{d}{dx}$ \\
& & & & \\
${A}_2$ & $\frac{x_2}{2} \frac{\partial}{\partial x_1}+\frac{x_1}{2} \frac{\partial}{\partial x_2}$ & $\ell \sin \theta + \cos \theta \frac{d}{d \theta}$  & $\ell x +\frac{1- x^2}{2} \frac{d}{dx}$ & $\frac{\ell}{2} (\frac{1}{x}+x) + \frac{1-x^2}{2} \frac{d}{dx}$ \\
& & & & \\
 ${N}_+$ & $x_1 \frac{\partial}{\partial x_2}$ & $\ell \sin \theta + (\cos \theta+1) \frac{d}{d \theta}$ & $2 \ell x -x^2 \frac{d}{d x}$  & $\ell x -x^2 \frac{d}{d x}$ \\
 & & & & \\
 ${N}_-$ & $x_2 \frac{\partial}{\partial x_1}$ & $\ell \sin \theta + (\cos \theta-1) \frac{d}{d \theta}$  & $\frac{d}{d x}$ & $\frac{\ell}{x}+\frac{d}{d x}$ \\
 & & & & \\
 \hline
\\[0.125cm]
\end{tabular}
\caption{Infinitesimal generators associated with the representation $T_\ell$ for various one-parameter subgroups.
The columns correspond to equivalent realisations of the representation in spaces of functions
on the punctured plane, the circle, the line and the hyperbola respectively.}
\label{infinitesimalGeneratorTable}
\end{table}

\subsection{One-dimensional disordered systems}
\label{disorderedSubsection}
The class of probability measures for which the transfer operator has a simple 
expression in terms of generators arises naturally in connection
with the one-dimensional disordered systems that were studied by Dyson \cite{Dy} and
Frisch \& Lloyd \cite{FL}. So
we provide a brief description of these models. The reader interested in their physical
relevance will find further details in \cite{CTT2,Te}.

Consider the differential equation
\begin{equation}
\notag
-\psi'' + U  \psi = \omega^2 M' \psi\,, \quad x > 0\,.
\end{equation}
Here, $U$ denotes the generalised function
$$
U (x) = \sum_{j} u_j \,\delta (x-x_j)
$$
where the $x_j$ are the values taken by a Poisson process of intensity $\rho$, so that the jumps
$$
l_j := x_{j+1}-x_j  > 0
$$
are independent and exponentially distributed with mean $1/\rho$.
If we put $M'(x) \equiv 1$ and draw the $u_j$ independently from some probability
distribution, we obtain a model for the energy
levels $\omega^2$ of a particle subject to a one-dimensional potential with impurities \cite{FL}.
By introducing the derivative $\psi'$ as an additional unknown, we can write this second-order differential equation as a first-order system.
For $\omega^2 =1$, the solution of the corresponding Cauchy problem may be expressed as a 
product $\Pi_n$ with elements
of the form
\begin{equation}
\notag
g_j = \underbrace{\begin{pmatrix}
\cos l_j & - \sin l_j \\
\sin l_j & \cos l_j
\end{pmatrix}}_{k(2 l_j) \in K}
\underbrace{\begin{pmatrix}
1 & u_j \\
0 & 1
\end{pmatrix}}_{n_+ (u_j) \in N_+}\,.
\end{equation}
The analysis carried out by Frisch \& Lloyd \cite{FL} and, later, Kotani \cite{Ko},
focused on the random process associated with the ``Riccati variable'' $\psi'/\psi$.

By setting $U \equiv 0$ and
$$
M'(x) = \sum_j m_j \, \delta (x-x_j)
$$
we obtain instead a model for a (discrete) vibrating string, where $m_j$ is the punctual
mass at $x_j$ and $\omega$ is now interpreted as a characteristic
frequency. The solution of the Cauchy problem in the case $\omega=1$
is then a product of matrices of the form
\begin{equation}
\notag
g_j = \underbrace{\begin{pmatrix}
1 & 0 \\
l_j & 1
\end{pmatrix}}_{n_-(l_j) \in N_-}
\underbrace{\begin{pmatrix}
1 & -m_j \\
0 & 1
\end{pmatrix}}_{n_+ (-m_j) \in N_+}\,.
\end{equation}
Dyson \cite{Dy} considered the case where the $l_j$ and $m_j$ are drawn from gamma distributions. We alert the reader to the fact that
our notation differs from Dyson's; his string equation is
$$
K_j \left ( x_{j+1}-x_{j} \right ) + K_{j-1} \left ( x_{j-1}-x_j \right ) = -m_j \omega^2 x_j
$$
where the $x_j$ are the positions of particles coupled together by springs that obey Hooke's law, and $K_j$ is the elastic modulus
of the spring between the $j$th and $(j+1)$th particles. The correspondence between 
this and our own equation is
$$
x_j \sim \psi (x_j) \;\;\text{and}\;\; K_j \sim 1/l_j\,.
$$

Other models of a similar kind have been considered more recently in which $M' \equiv 1$
and the potential function $U$ is ``supersymmetric'' \cite{CTT}:
$$
U = \frac{w'}{2} + \frac{w^2}{4}\,,\;\;w(x) := \sum_j w_j \,\delta (x-x_j)\,.
$$
These models lead to products of matrices of the form
\begin{equation}
\notag
g_j = \underbrace{\begin{pmatrix}
\cos l_j & - \sin l_j \\
\sin l_j & \cos l_j
\end{pmatrix}}_{k(2 l_j) \in K}
\underbrace{\begin{pmatrix}
e^{\frac{w_j}{2}} & 0 \\
0 & e^{-\frac{w_j}{2}}
\end{pmatrix}}_{a_1 (w_j) \in A_1}
\end{equation}
or
\begin{equation}
\notag
g_j = \underbrace{\begin{pmatrix}
\text{ch} \,l_j & \text{sh} \,l_j \\
\text{sh} \,l_j & \text{ch} \,l_j
\end{pmatrix}}_{a_2(2 l_j) \in A_2}
\underbrace{\begin{pmatrix}
e^{\frac{w_j}{2}} & 0 \\
0 & e^{-\frac{w_j}{2}}
\end{pmatrix}}_{a_1 (w_j) \in A_1}\,.
\end{equation}

All these models have two features in common that make the corresponding transfer operator
tractable: Firstly, the matrices consist of two {\em independent} components
taken from two of the one-parameter subgroups, say $D$ and $E$. This implies that the transfer operator {\em factorises}: Indeed, for
$$
g_j = d (t_j)\, e(\tau_j)
$$ 
we may write
\begin{equation}
\notag
{\mathscr T}_\ell = {\mathbb E} \left [ T_\ell \left ( g_j \right ) \right ] 
= {\mathbb E} \left \{ T_\ell \left [ d (t_j) \right ] \right \} \,{\mathbb E} \left \{ T_\ell \left [ e(\tau_j) \right ] \right \}
= {\mathbb E} \left ( e^{t_j {D}}\right ) \,{\mathbb E} \left ( e^{\tau_j {E}}\right ) 
\end{equation}
where, with some abuse of notation, ${D}$ and ${E}$ are the infinitesimal generators associated with each of the subgroups. 
The second important feature is that one of the random parameters involved in the definition
of $g_j$ has an exponential or, more generally, a gamma distribution. 
For definiteness, suppose that
$$
{\mathbb P} \left [ \tau_j >  \tau \right ] = e^{-\rho \tau}\,.
$$
It then follows that
$$
{\mathbb E} \left ( e^{\tau_j {E}} \right ) = \left ( 1- \frac{1}{\rho} {E} \right )^{-1}
$$
and there only remains to deal with the component corresponding to the other
one-parameter subgroup--- in this case $D$. We can do so by working with a realisation
of the representation $T_\ell$ that is {\em diagonal} with respect to that subgroup.
In such a realisation, the infinitesimal generator associated with $D$ is independent 
of $\ell$ and
acts on the representation space by multiplication by a fixed function, so that, after
averaging, the $d$ component
is essentially the characteristic function of the random variable $t_j$. Hence, since $E$ is an affine function of $\ell$, by working
with the {\em inverse} ${\mathscr T}_\ell^{-1}$ we
obtain a spectral problem equivalent to that for the transfer operator, in terms of a second-order difference/differential operator that depends linearly
on the index $\ell$.

The utility of this approach depends on finding realisations
of the representation $T_\ell$ that are diagonal with respect to a given subgroup. This matter
is discussed in considerable detail in Vilenkin's book \cite{Vi} and in its sequel \cite{VK1}; we shall draw heavily on that material. The upshot
is that we can--- at least for three of the subgroups--- find such realisations, all of them equivalent to the basic realisation
associated with even homogeneous functions on the punctured plane.
In order to guide the reader through this collection, we have
organised it in terms of a geometrical theme, based on particular curves that serve as parametrisations
of the Furstenberg boundary; see Figure \ref{fig:overview}.
To each curve is associated a
functional transform that diagonalises the representation with respect to some subgroup. 

\subsection{The finite-dimensional case}
\label{finiteDimensionalSubsection}
For $\ell \in {\mathbb Z}$, the representation $T_\ell$ is reducible. One of the invariant subspaces is
\begin{equation}
V_\ell^0 := \underset{-\ell \le n \le \ell}{\text{span}} \{ {{\mathbf e}_{\ell,n}} \}\,.
\label{finiteDimensionalInvariantSubspace}
\end{equation}
Now, the function 
$1_\ell$, defined by Equation (\ref{unitFunction}) and appearing in Equation (\ref{representationFormula}), happens to belong to $V_\ell^0$; it may therefore be decomposed in terms of the eigenvectors of the restriction of
${\mathscr T}_\ell^{-1}$ to $V_\ell^0$. It follows that the generalised Lyapunov exponent
can, for $\ell \in {\mathbb Z}$, be obtained by purely algebraic
means--- a fact that has been observed in the physics literature \cite{ST1,Va,ZP} without reference to
representation theory.
For example, the smallest eigenvalue of ${\mathscr T}_0^{-1}$ is readily deduced by restricting this operator to the one-dimensional invariant 
subspace
$$
V_0^0 := \text{span} \{ {\mathbf e}_{0,0} \}\,.
$$
The restriction has only one eigenvalue, namely $\lambda= \lambda_0 := 1$. 
We choose a non-zero multiple of ${\mathbf e}_{0,0}$ and denote it by $v_0$.

\subsection{The adjoint spectral problem and the Dyson--Schmidt equation}
\label{dysonSchmidtSubsection}
In the general case, it will be more convenient to work
with the {\em adjoint} 
\begin{equation}
{\mathscr A}_\ell  := \left ( {\mathscr T}_\ell^{-1} \right )^\ast
\label{adjointOfInverse}
\end{equation}
of ${\mathscr T}_\ell^{-1}$.
We shall see that this adjoint operator is defined in the space $V_{\ell^\ast}$
where
\begin{equation}
\ell^\ast = -\overline{\ell}-1
\label{ellStar}
\end{equation}
and the bar denotes complex conjugation.
The relationship between the indices $\ell$ and $\ell^\ast$ will be of great significance for what follows.
We denote by
$\braket{\cdot |\cdot}$ a positive-definite Hermitian form expressing the duality between
the spaces $V_\ell$ and $V_{\ell^\ast}$, such that
$\{ {\mathbf e}_{\ell^\ast,n} \}_{n \in {\mathbb Z}} \subset V_{\ell^\ast}$ and $\{ {{\mathbf e}_{\ell,n}} \}_{n \in {\mathbb Z}} \subset V_\ell$ constitute a bi-orthogonal system: 
$$
\braket{{\mathbf e}_{\ell^\ast,m} | \,{\mathbf e}_{\ell, n}} = \begin{cases}
1 & \text{if $m=n$} \\
0 & \text{otherwise}
\end{cases}\,.
$$
The adjoint spectral problem then takes the form: Find $\lambda \in {\mathbb C}$ and 
$0 \ne f \in V_{\ell^\ast}$ such that
\begin{equation}
\overline{\lambda} \,f = {\mathscr A}_\ell f = \left ( {\mathscr A}_0 + \overline{\ell}\,
{\mathscr B} \right ) f\,.
\label{adjointSpectralProblem}
\end{equation}
Here,  we have used the fact, discussed
at the end of \S \ref{disorderedSubsection}, that for the products we consider, we can work with a realisation such that ${\mathscr T}_\ell^{-1}$ is an affine function of $\ell$, so that
its adjoint ${\mathscr A}_\ell$ is an affine function of $\overline{\ell}$. 
More precisely, we shall be interested in computing the {\em smallest} eigenvalue
$\overline{\lambda}$. By virtue of Equation (\ref{adjointOfInverse}), $1/\lambda$ is then the largest eigenvalue of the transfer operator, and so we can access the generalised Lyapunov exponent via the formula
\begin{equation}
  \label{eq:RelationLambdaSmallBig}
  \lambda = \exp[-\Lambda(2\ell)]
  \:.
\end{equation}

For $\ell=0$, we have $\lambda =\lambda_0 = 1$, and so Equation (\ref{adjointSpectralProblem}) becomes
\begin{equation}
\left ( {\mathscr A}_0 - {I} \right ) f_0 = 0\,.
\label{dysonSchmidtEquation}
\end{equation}
This is called {\em the Dyson--Schmidt equation} in the physics literature and it has a useful
interpretation in terms of the Furstenberg theory: As mentioned already,
under certain conditions on the distribution of the matrices in the product,
there exists a probability measure on the Furstenberg boundary, invariant under the action of the
matrices drawn at random from the group. One may think of $f_0$ as the ``density'' of the
invariant measure with respect to some fixed measure associated with the particular
realisation. With some abuse of terminology, we shall
henceforth refer to $f_0$ as the {\em invariant density} associated with the product,
even though the invariant measure may be singular or merely continuous.

\subsection{Perturbative solution about $\ell=0$}
\label{perturbativeSubsection}
It is apparent from the explicit form of the generators that ${\mathscr A}_\ell$
is a regular perturbation of ${\mathscr A}_0$.
If there exists an invariant density $f_0$, we can,
for small values of the index $\ell$, look for a solution of the adjoint spectral problem 
(\ref{adjointSpectralProblem}) of the form
\begin{equation}
\lambda = 1 + \lambda_1 \ell + \lambda_2 \ell^2 + \cdots
\;\;\text{and}\;\;
f = f_0 + f_1 \overline{\ell} + f_2 \overline{\ell^2} + \cdots\,.
\label{perturbationExpansion}
\end{equation}
The problem is then identical to that considered in \cite{VL}, Chapter 1, \S 3.
Upon equating like powers of $\overline{\ell}$, we obtain the recurrence relation
\begin{equation}
\left ( {\mathscr A}_0 - {I} \right ) f_{j} = \sum_{i=1}^j \overline{\lambda_i} f_{j-i} - {\mathscr B} \,f_{j-1}\,,\;\; j \in {\mathbb N}\,,
\label{recurrenceRelation}
\end{equation}
with the starting value $f_{-1} = 0$. We choose the normalisation 
\begin{equation}
\braket{f_j | \,v_0}  = \begin{cases}
1 & \text{if $j=0$} \\
0 & \text{otherwise}
\end{cases}\,.
\label{perturbativeNormalisationCondition}
\end{equation}

The effective solution of Equation (\ref{recurrenceRelation})
depends to some extent on the details of the particular 
realisation of the representation $T_\ell$, but the guiding principles may be described as follows:
The operator ${\mathscr A}_0-{I}$ has a non-trivial kernel and so, for a solution
$f_j$ of Equation (\ref{recurrenceRelation}) to exist, the right-hand side cannot be arbitrary:
For the simplest realisation--- that on the circle--- we require that the right-hand side
of Equation (\ref{recurrenceRelation}) be bi-orthogonal to $v_0$, the eigenvector
of ${\mathscr T}_0^{-1}$ corresponding to the eigenvalue $\lambda_0=1$; see \cite{VL}.
In view of the normalisation (\ref{perturbativeNormalisationCondition}),
this yields
\begin{equation}
\lambda_j  = \braket{{\mathscr B}  f_{j-1} |\, v_{0}} \,.
\label{eigenvalueFormula}
\end{equation}
With this formula for $\lambda_j$, we may then write
\begin{equation}
\notag
\sum_{i=1}^j \overline{\lambda_i} f_{j-i} - {\mathscr B} \,f_{j-1} = r_j
\end{equation}
where 
$$
r_j \in  V_0^\perp = \underset{n \ne 0}{\text{span}} \{ {\mathbf e}_{-1,n} \} =
\underset{n < 0}{\text{span}} \{ {\mathbf e}_{-1,n} \} \cup \underset{n > 0}{\text{span}} \{ {\mathbf e}_{-1,n} \}\,.
$$
These last two subspaces of $V_{-1}$ are invariant under ${\mathscr A}_0 - {I}$, and it turns out  that $f_j$ may be obtained by considering the restriction of Equation
(\ref{recurrenceRelation}) to any one of them. Intuitively, one might think of this 
simplification as a consequence of the fact that, in the case $\ell=0$, there are two subrepresentations of $T_\ell$
that belong to the discrete series. 

For some realisations, a complication arises
because the term ${\mathscr B} f_{j-1}$ appearing 
in Equation (\ref{eigenvalueFormula}) does not belong to $V_{-1}$, so that the inner product with $v_0$ is not well-defined. We shall see that, in such cases, a natural regularisation
may be used to give a precise meaning to Equation (\ref{eigenvalueFormula}). By using this machinery, we can in principle determine the successive terms in the expansions for 
$\lambda$. In particular, this translates into expressions for the growth rate $\gamma$ of the product and the variance 
$\sigma^2$ of its fluctuations via the formulae
\begin{equation}
\gamma = -\frac{\lambda_1}{2} \;\;\text{and}\;\;\sigma^2 = \frac{\lambda_1^2}{4} - \frac{\lambda_2}{2}\,.
\label{growthRateAndVarianceInTermsOfLambda}
\end{equation}
We remark that this formula
for $\gamma$ is much simpler than the the usual Furstenberg formula, which also involves an integral over the group \cite{Fu}. This simplification is of course specific to the probability
distributions on the group that we consider in this paper.

\subsection{Outline of the remainder of the paper}

\S \ref{circleSection} is devoted to the realisation associated with the unit circle of equation
$$
x_1^2 + x_2^2 = 1
$$ 
in which the representation
space is in fact independent of $\ell$ and consists of smooth $2 \pi$-periodic functions of a real variable. We shall use this realisation to develop the theoretical framework in greater detail, with the aim of formulating the spectral problem for the transfer operator
in concrete terms.  By going over
to Fourier series, one can diagonalise with respect to the subgoup $K$ and this results in a spectral problem involving a difference operator. 
As a practical application, we study in 
\S \ref{KNPlusSection} the growth and fluctuations of products of the type $k \,n_+$
where the components are independent and the $n_+$ component is exponentially distributed.
In \S \ref{lineSection}, we consider a realisation associated with the line
of equation
$$
x_2 = 1\,.
$$
It is that realisation that is associated with the Riccati analysis
mentioned earlier. The representation
space then consists of certain smooth functions on the projective line, whose
behaviour at infinity depends on the index $\ell$. By going over to the Fourier transform, we achieve a representation 
diagonal with respect to the subgroup $N_-$. 
In this realisation, the spectral problem involves 
a differential operator. 
In \S \ref{NminusSection}, we study the growth and fluctuations of products involving
the subgroup $N_-$.  We examine products of the type 
$n_- \,k$, as such products correspond to the model considered by Frisch \& Lloyd.
We then consider products  of the type $n_- \,n_+$ associated with Dyson's model.
We also take the opportunity to discuss--- albeit briefly---
the existence of the invariant measure.
In \S \ref{hyperbolaSection}, we discuss a realisation 
relevant to disordered systems of ``supersymmetric'' type; it is associated with the hyperbola
of equation
$$
x_2 = \frac{1}{|x_1|}\,.
$$
Apart from a universal
factor, the functions in its representation space are essentially the same as for the realisation on the line.
By using the Mellin transform, we can diagonalise with respect to the subgroup $A_1$, and formulate the spectral problem in terms of a difference operator analogous to those discussed recently
by Neretin \cite{Ne}. It turns out, however, that the study of this operator is not entirely straightforward, and we shall indicate the nature of the difficulties.
Finally, in \S \ref{conclusionSection}, we discuss the possible extension of these ideas to other semi-simple groups.

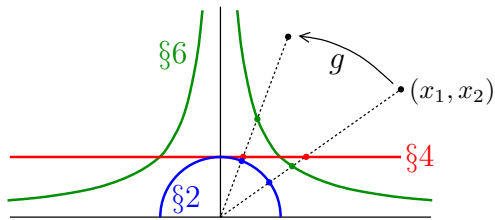
\begin{figure}[!ht]
\centering
\begin{picture}(0,0)%
\includegraphics{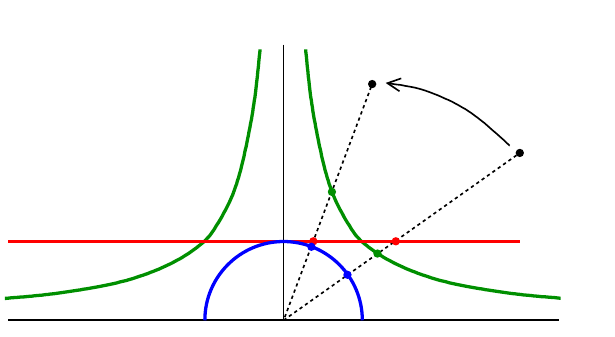}%
\end{picture}%
\setlength{\unitlength}{4144sp}%
\begingroup\makeatletter\ifx\SetFigFont\undefined%
\gdef\SetFigFont#1#2#3#4#5{%
  \reset@font\fontsize{#1}{#2pt}%
  \fontfamily{#3}\fontseries{#4}\fontshape{#5}%
  \selectfont}%
\fi\endgroup%
\begin{picture}(2803,1534)(235,-1757)
\put(2182,-788){\makebox(0,0)[lb]{\smash{{\SetFigFont{12}{14.4}{\rmdefault}{\mddefault}{\updefault}{\color[rgb]{0,0,0}$g$}%
}}}}
\put(2660,-966){\makebox(0,0)[lb]{\smash{{\SetFigFont{10}{12.0}{\rmdefault}{\mddefault}{\updefault}{\color[rgb]{0,0,0}$(x_1,x_2)$}%
}}}}
\put(1156,-760){\makebox(0,0)[lb]{\smash{{\SetFigFont{12}{14.4}{\rmdefault}{\mddefault}{\updefault}{\color[rgb]{0,.56,0}\S6}%
}}}}
\put(1237,-1620){\makebox(0,0)[lb]{\smash{{\SetFigFont{12}{14.4}{\rmdefault}{\mddefault}{\updefault}{\color[rgb]{0,0,1}\S2}%
}}}}
\put(2634,-1378){\makebox(0,0)[lb]{\smash{{\SetFigFont{12}{14.4}{\rmdefault}{\mddefault}{\updefault}{\color[rgb]{1,0,0}\S4}%
}}}}
\end{picture}%
\caption{Curves in the punctured plane associated with various realisations
of $T_\ell$, and the section of the paper in which each is discussed. The (semi-) circle, the horizonal line and the two-branched hyperbola lead to realisations diagonal with respect to the subgroups $K$, $N_-$ and $A_1$ respectively.}
\label{fig:overview}
\end{figure}

The reader should be aware that the material presented in \S\S \ref{circleSection},  \ref{lineSection} and \ref{hyperbolaSection} is largely a distillation of that in \cite{Vi,VK1}, adapted
to our particular purpose.

\section{Realisation on the circle}
\label{circleSection}
In this section, we shall obtain and work with a realisation of $T_\ell$
that is diagonal with respect to the subgroup $K$. Although this realisation
is admittedly of limited interest from the point of view of disordered systems, 
it has the advantage of affording the most natural setting 
in which to develop the technical details of our approach.

Consider the half-circle of equation
$$
x_1^2 + x_2^2 =1\,,\;\;x_2 > 0\,.
$$ 
Using
the polar decomposition
$$
{\mathbf x} = r \begin{pmatrix} \cos \theta & \sin \theta \end{pmatrix}\,,\;\; 0 \le \theta \le \pi\,,
$$
we see that, for every $v \in V_\ell$, homogeneity implies
$$
v({\mathbf x}) = r^{2 \ell} v(\cos \theta, \sin \theta) 
$$
whilst evenness implies
$$
v( \cos ( \theta + \pi),\sin (\theta+\pi)) = v(\cos \theta, \sin \theta)\,.
$$
Thus, $v$ is completely determined by a smooth $\pi$-periodic function. Conversely, to any smooth $\pi$-periodic function, say $\widetilde{v}$, we can associate the function $v :\,{\mathbb R}_\ast^2 \rightarrow {\mathbb C}$ defined by
$$
v({\mathbf x}) = r^{2\ell} \widetilde{v} (\theta)
$$
where $r$ and $\theta$ are the polar coordinates of ${\mathbf x}$. Obviously, the function $v$ defined in this way is smooth, even and homogeneous of degree $2 \ell$.
It will be convenient for esthetic 
reasons to rescale the angle so that we associate with $v$ --- not a $\pi$-periodic function--- but rather a $2 \pi$-periodic function
$\widetilde{v}$ defined by
$$
\widetilde{v}(\theta) = v \left ( \cos \frac{\theta}{2}, \sin \frac{\theta}{2} \right )\,,\;\;0 \le \theta \le 2 \pi\,.
$$

We denote the map $v \mapsto \widetilde{v}$ by $Q$ and write
\begin{equation}
\widetilde{V}_{\ell} = \left \{ Q v :\, v \in V_\ell \right \}\,.
\label{realisationSpace}
\end{equation}
For example, the function $1_\ell \in V_\ell$, defined by Equation (\ref{unitFunction}), is mapped by $Q$ to the
$2 \pi$-periodic function that is identically
equal to $1$--- hence the notation. 

\begin{proposition}
For a function $\widetilde{v} : [0, 2 \pi) \rightarrow {\mathbb C}$ to belong to $\widetilde{V}_\ell$, it is necessary
and sufficient that its $2 \pi$-periodic extension to the real line be smooth.
\label{circleProposition}
\end{proposition}

The map $Q$ from $V_\ell$ to $\widetilde{V}_\ell$ is bijective. 
We may then realise the representation $T_\ell$ on the space of smooth $2\pi$-periodic functions as follows:
\begin{equation}
\widetilde{T}_\ell(g)  \,Q  = Q \,T_\ell(g) \,.
\label{intertwinning}
\end{equation}
More explicitly,
\begin{multline}
\left [ \widetilde{T}_\ell(g) \widetilde{v} \right ] (\theta) := \left [ T_\ell(g) v \right ] \left ( \cos \frac{\theta}{2}, \sin \frac{\theta}{2} \right ) \\
= v \left ( a \cos \frac{\theta}{2} + c \sin \frac{\theta}{2}, \,b \cos \frac{\theta}{2} + d \sin \frac{\theta}{2} \right ) 
= \left | \begin{pmatrix} \cos \frac{\theta}{2} & \sin \frac{\theta}{2} \end{pmatrix} g \right |^{2 \ell}\,
 \widetilde{v}(\theta \cdot g)
\label{circleRepresentation}
\end{multline}
where $\theta \cdot g$, defined implicitly by
\begin{equation}
\cot \frac{\theta \cdot g}{2} := \frac{a \cos \frac{\theta}{2} + c \sin \frac{\theta}{2}}{b \cos \frac{\theta}{2}+d \sin \frac{\theta}{2}}\,,
\label{actionOnTheCircle}
\end{equation}
describes the action of the group on the circle.
We shall henceforth refer to this as the {\em realisation on the circle}
and drop the tilde. The corresponding infinitesimal generators are listed in the third column of Table \ref{infinitesimalGeneratorTable}.

\subsection{Multipliers}
There is an alternative interpretation of this (and other) realisatisation(s).
We observe that the measure $d \theta$ is invariant under the action of the subgroup $K$. 
For a fixed $g \in G$, the Jacobian of the map $\theta \mapsto \theta \cdot g$ is easily computed:
\begin{equation}
\frac{d}{d \theta} \left (  \theta \cdot g \right ) = \frac{1}{\left (a \cos \frac{\theta}{2} + c \sin \frac{\theta}{2} \right )^2 + \left (b \cos \frac{\theta}{2} + d \sin \frac{\theta}{2} \right )^2}
=: \sigma_K (\theta,g)\,.
\label{circleMultiplier}
\end{equation}
We may therefore express the operator $T_\ell(g)$ as
\begin{equation}
\left [ T_\ell(g) v \right ] (\theta) = \sigma_K(\theta,g)^{-\ell} \,v ( \theta \cdot g )\,.
\label{circleRealisationWithMultiplier}
\end{equation}
Now, since the map
$\sigma_K : \,[0,2\pi) \times G \rightarrow {\mathbb R}_+$  so defined is a Jacobian, it follows easily from the chain rule for differentiation that it has the so-called {\em multiplicative cocycle}
or {\em multiplier} property :
\begin{equation}
\forall\, \theta \in [0,2 \pi)\,,\;\;\forall\, g_1,\,g_2 \in G\,,\;\; \sigma_K ( \theta, g_1 g_2) = \sigma_K (\theta \cdot g_1,g_2) \,\sigma_K (\theta,g_1)\,.
\label{multiplicativeCocycleForTheCircle}
\end{equation}
Using this property, we may then verify directly, without reference to the realisation on the punctured plane, that the realisation on the circle is indeed a representation. This  formulation in terms of multipliers accounts for the fact, noted in \S \ref{generatorSubsection}, that the infinitesimal generators are affine functions of the index $\ell$. We note, for future reference, another easy consequence of the multiplier property:
\begin{equation}
1 = \sigma_K (\theta, e) = \sigma_K (\theta,g g^{-1}) = \sigma_K ( \theta \cdot g,g^{-1}) \,\sigma_K (\theta,g)\,.
\label{anotherMultiplierProperty}
\end{equation}

\subsection{Hilbert space formulation}
The monomials
\begin{equation}
{{\mathbf e}_{\ell,n}}(\theta) := e^{\i n \theta}
\label{circleBasis}
\end{equation}
form an obvious basis for $V_\ell$, orthonormal with respect to the inner product
\begin{equation}
\braket{f |\, v} := \frac{1}{2 \pi} \int_{0}^{2 \pi} \overline{f(\theta)} \,v(\theta)\,d \theta\,.
\label{innerProductOnTheCircle}
\end{equation}
We denote by $H$ the Hilbert space obtained from $V_\ell$ by completion
with respect to the norm induced by this inner product; it coincides with the familiar Hilbert space 
of square integrable functions on the circle, and we call the numbers 
$\braket{{\mathbf e}_{\ell,n} |\, v}$
the {\em Fourier coefficients} of $v$.
$V_\ell$ is then the subspace of $H$ consisting of functions $v$
whose Fourier coefficients are {\em rapidly decreasing}, i.e.
$$
\forall\, m \in {\mathbb N}\,,\;\;n^m \braket{{\mathbf e}_{\ell,n} | \,v} \xrightarrow[n \rightarrow \infty]{} 0\,.
$$

\subsection{The dual space and the adjoint}
It is possible to define a topology on $V_\ell$ so that one may speak of continuous linear functionals on $V_\ell$ \cite{LN}. 
The space $V_\ell^\ast$ of these functionals is the {\em dual} of $V_\ell$ and may be identified with the set of functions $f$ whose Fourier coefficients are
{\em slowly increasing}, i.e.
$$
\exists \, m \in {\mathbb N}\,,\; \exists\,c >0\;\;\text{such that}\;\;\forall\, n \in {\mathbb N}\,,\;\;
\left | \braket{f |\, {\mathbf e}_{\ell,n}} \right | \le c \,n^m\,.
$$
Fourier series with such coefficients always converge in the sense of generalised functions; see \cite{GS}.
Clearly,
$$
V_\ell \subset H \subset V_\ell^\ast\,.
$$
With some abuse of notation, we may then write $\braket{f |\, v}$ to denote also the value that the continuous linear functional
takes when it is evaluated at $v \in V_\ell$.

Let $A$ be a linear operator on $V_\ell$. Its {\em adjoint} $A^\ast$ is a linear operator on $V_\ell^\ast$ that assigns to $f \in V_\ell^\ast$
the linear functional $A^\ast f$ defined via
\begin{equation}
\forall\,v \in V_\ell\,,\;\;\braket{A^\ast f |\, v}= \braket{f |\, A v} \,.
\label{adjointDefinition}
\end{equation}
Let us work out the adjoint of $T_\ell(g)$. For the realisation on the circle, we have
$$
\braket{f |\, T_\ell(g) v} = \frac{1}{2 \pi} \int_0^{2 \pi} \overline{f (\theta)}\, \sigma_K(\theta,g)^{-\ell} \,v (\theta \cdot g)\,d \theta
$$
where $\sigma_K$ is the multiplier defined by (\ref{circleMultiplier}). Hence, the change of variable $\widetilde{\theta} = \theta \cdot g$ yields
\begin{multline}
\braket{f | \, T_\ell(g) v} = \frac{1}{2 \pi} \int_0^{2 \pi} \overline{f (\widetilde{\theta} \cdot g^{-1})}\, \sigma_K(\theta,g)^{-\ell-1} \,v (\widetilde{\theta})\,d \widetilde{\theta} \\
= \frac{1}{2 \pi} \int_0^{2 \pi} \overline{f (\widetilde{\theta} \cdot g^{-1})}\, \sigma_K(\widetilde{\theta},g^{-1})^{\ell+1} \,v (\widetilde{\theta})\,d \widetilde{\theta} \,.
\notag
\end{multline}
where we have used the definition (\ref{circleMultiplier}) of $\sigma_K$ and its property (\ref{anotherMultiplierProperty}). We readily deduce the formula
\begin{multline}
\left [ T_\ell^\ast (g) f \right ] (\theta) = \sigma_K (\theta,g^{-1})^{-\ell^\ast} \,f (\theta \cdot g^{-1}) \\
= \left [ \left ( d \cos \frac{\theta}{2} - c\sin \frac{\theta}{2} \right )^2 + \left ( -b \cos \frac{\theta}{2} +a \sin \frac{\theta}{2} \right )^2\right ]^{\ell^\ast}
\, f \left ( \theta \cdot g^{-1} \right )
\label{adjointOnTheCircle}
\end{multline}
where 
$$
\cot \frac{\theta \cdot g^{-1}}{2} = \frac{d \cos \frac{\theta}{2} - b \sin \frac{\theta}{2}}{a \cos \frac{\theta}{2} - c \sin \frac{\theta}{2}}
$$
and $\ell^\ast$ is the index (\ref{ellStar}).

It follows from this calculation that
\begin{equation}
T_{\ell}^\ast (g)= T_{\ell^\ast}(g^{-1}) \,.
\label{adjointProperty}
\end{equation}
In particular, we deduce that $V_{\ell^\ast}$ is a subspace of $V_\ell^\ast$, invariant
for the adjoint $T_\ell^\ast(g)$.
Knowing the infinitesimal generators, we can easily construct the adjoint by using the
\begin{proposition}
Let $S = S(\ell)$ be an infinitesimal generator of ${T}_\ell$ corresponding to
some one-parameter subgroup. Then its adjoint is given by the formula
$$
S(\ell)^{\ast} = - S (\ell^\ast)\;\;\text{where}\;\;\ell^\ast = -\overline{\ell}-1\,.
$$
\label{firstAdjointProposition}
\end{proposition}

Another consequence of the formula for the adjoint is as follows:
The spectrum of ${\mathscr T}_{\ell}$ is the complex conjugate of the spectrum of its adjoint
$$
{\mathscr T}_{\ell}^\ast = {\mathbb E} \left [ T_\ell^\ast (g) \right ]  = {\mathbb E} \left [ T_{\ell^\ast} (g^{-1}) \right ]\,. 
$$
From the definition of $\Lambda_{\ell}$, we see that the largest eigenvalue of this last operator is also the largest eigenvalue of the operator
$$
{\mathbb E} \left [ T_{\ell^\ast} ( g^{-t}) \right ] \,.
$$
But this operator is similar to ${\mathscr T}_{\ell^\ast}$ because, for every $g \in G$,
$$
\begin{pmatrix}
0 & 1 \\
-1 & 0
\end{pmatrix} g = g^{-t}  \begin{pmatrix}
0 & 1 \\
-1 & 0
\end{pmatrix}\,.
$$
We deduce 
\begin{equation}
\Lambda (2 \ell) = \overline{\Lambda(2 \ell^\ast)} = \Lambda \left (\overline{2 \ell^\ast} \right ) = \Lambda (-2\ell-2)\,.
\label{symmetryProperty}
\end{equation}
This well-known symmetry property of the generalised Lyapunov exponent--- see \cite{Va}, Proposition 2--- is consistent with the fact that the representations
$T_\ell$ and $T_{-\ell-1}$ are equivalent.


\subsection{The solvable case $\ell \in {\mathbb Z}$}
\label{discreteSection}
By virtue of the symmetry property (\ref{symmetryProperty}), we need only consider the
range
$$
\ell \in \left \{ 0, 1,\,\ldots \right \}\,.
$$
We can use the 
formulae (\ref{eigenvectorsOfJ0}-\ref{raisingLowering}) to show  
that $V_\ell^0$, defined by Equation (\ref{finiteDimensionalInvariantSubspace}),
and
\begin{equation}
\notag
\underset{\pm n \ge -\ell}{\text{span}} \{ {{\mathbf e}_{\ell,n}} \}\,.
\end{equation}
are subspaces invariant under $T_\ell$. 
The finite-dimensional subspace $V_\ell^0$
contains no proper invariant subspace. Hence the subrepresentation $T_\ell^0$ obtained
by restricting $T_\ell$
to $V_\ell^0$ is irreducible.
\begin{proposition}
For $\ell \in {\mathbb N}$, the generalised Lyapunov exponent $\Lambda(2 \ell)$ is the logarithm of the largest
eigenvalue of the $(2\ell+1) \times (2 \ell+1)$ matrix with entries 
\begin{equation}
\braket{{\mathbf e}_{\ell^\ast, m} | {\mathscr T}_\ell \, {\mathbf e}_{\ell,n}}\,,\;\;-\ell \le m,\,n \le \ell\,.
\label{finiteDimensionalSubrepresentation}
\end{equation}
\label{solvabilityProposition}
\end{proposition}
\begin{proof}
The finite-dimensional
subrepresentation $T_\ell^0$ is expressible as a matrix of order $2\ell+1$ with entries
$$
\braket{{\mathbf e}_{\ell^\ast, m} | \,{T}_\ell (g) \,{\mathbf e}_{\ell,n}}\,,\;\;-\ell \le m,\,n \le \ell\,.
$$
After averaging it over the group, we obtain the matrix with the stated entries. It only remains
to observe that 
$$
1_\ell = \mathbf{e}_{\ell,0} \in V_\ell^0
$$
and so it may be expressed in terms of the eigenvectors (and associated eigenvectors) of the averaged matrix. In particular, 
${\mathscr T}_\ell^n 1_\ell$ will be asymptotic in magnitude to the largest eigenvalue, raised to the power $n$.
\end{proof}

Two further irreducible subrepresentations, denoted
$T_\ell^\pm$, are obtained
by restricting $T_\ell$ to the quotient spaces
\begin{equation}
\underset{\pm n \ge -\ell}{\text{span}} \{ {{\mathbf e}_{\ell,n}} \}/V_\ell^0\,.
\label{infiniteDimensionalInvariantSubspaces}
\end{equation}
Here, taking the quotient means that two elements
are considered identical if their difference belongs to $V_\ell^0$.
Thus an element in $V_\ell^\pm$ is completely determined
by its coefficients of index $\pm n > \ell$. 

The significance of the subspaces $V_\ell^{\pm}$ is as follows: For the realisation on the circle, we may identify $v \in V_\ell^+$ with an element of $V_\ell$ of the form
$$
v = \sum_{n > \ell} \braket{{\mathbf e}_{\ell,n} | v} e^{\i n \theta}\,.
$$ 
Since the coefficients $\braket{{\mathbf e}_{\ell,n} | v}$ are rapidly decreasing, we can assert that, 
for every $r \le 1$, the power series
$$
\sum_{n > \ell} \,\braket{{\mathbf e}_{\ell,n} | v} z^n\,,\;\;\text{with $z = r e^{\i \theta}$}\,,
$$
is convergent, and so the series is a function of the complex variable $z$, analytic inside the unit disk. We may therefore paraphrase our description of $V_\ell^+$
by saying that it consists of functions in  $V_\ell$, defined on the unit circle in the complex plane, that may be continued analytically inside the unit disk. Clearly, however, it does not contain all such functions.
Likewise, the subspace $V_\ell^-$ consists of functions that may be continued analytically {\em outside} the unit disk. For these reasons, we shall henceforth refer to $T_\ell^+$ and $T_\ell^-$
as the {\em holomorphic representations}. We will see in due course 
that they are {\em unitary}
with respect to a certain inner product.

\subsection{Diagonalisation with respect to $K$}
\label{discreteFourierSubsection}
It is straightforward to diagonalise $T_\ell$ with respect to the subgroup $K$ by working in the space
$$
\widehat{V}_\ell := \left \{ Q v : \, v \in V_\ell \right \}
$$
where $Q$ is the operator that assigns to $v \in V_\ell$ the sequence 
$\widehat{v} :\,{\mathbb Z} \rightarrow {\mathbb C}$ of its Fourier coefficients defined by
$$
\widehat{v}(n) := \braket{{\mathbf e}_{\ell,n} | \,v}\,.
$$ 
Thus, for instance,
$\widehat{{\mathbf e}}_{\ell,n}$ is the sequence 
defined by
$$
\widehat{{\mathbf e}}_{\ell,n} (m) = \delta_{m,n}\;\;\text{for every $m \in {\mathbb Z}$}\,.
$$

As mentioned already, $\widehat{V}_\ell$ is the space of rapidly decreasing sequences,
and its dual $\widehat{V}_\ell^\ast$ is the space of slowly increasing sequences. By virtue of
the Parseval formula, it follows from the definition (\ref{innerProductOnTheCircle}) of the inner product in $H$ that the duality
can be expressed via the positive-definite Hermitian form
$$
\braket{\widehat{f} \,|\, \widehat{v}} := \sum_{n \in {\mathbb Z}} \overline{\widehat{f}(n)}\,{\widehat{v}(n)}\,.
$$
The corresponding Hilbert space consists of the square summable sequences.
The $\widehat{{\mathbf e}}_{\ell^\ast,m}$ and  $\widehat{{\mathbf e}}_{\ell,n}$ form a bi-orthogonal system with respect to this inner product, 
but since,  for the realisation on the circle, there holds ${\mathbf e}_{\ell^\ast,n} = {\mathbf e}_{\ell,n}$, we have in fact an orthonormal basis.

We obtain a realisation--- denoted $\widehat{T}_\ell$--- of $T_\ell$ in
$\widehat{V}_\ell$
via
$$
\widehat{T}_\ell \,Q = Q \,T_\ell
$$
or, equivalently,
$$
\left [ \widehat{T}_\ell (g) \,\widehat{v} \right ] (n) = \braket{\widehat{\mathbf e}_{\ell^\ast,n} | \,T_\ell (g) v}\,.
$$
In this realisation, for every $g \in G$, $\widehat{T}_\ell(g)$ is a ``difference operator'' on the space of rapidly decreasing sequences; it may be identified with the infinite matrix with entry
$$
\braket{{\mathbf e}_{\ell^\ast,m} | \,{T}_\ell(g) \,{\mathbf e}_{\ell,n}} 
$$
in the $m$th row and $n$th column.
As an illustration, when 
$$
g = k(t) = \begin{pmatrix}
\cos \frac{t}{2} & \sin \frac{t}{2} \\
-\sin \frac{t}{2}. & \cos \frac{t}{2}
\end{pmatrix}
\in K,
$$
we find
$$
\left [ \widehat{T}_\ell (g) \,\widehat{v} \right ] (n) = e^{-\i n t} \,\widehat{v}(n)
$$
so that, for $g \in K$, $\widehat{T}_{\ell} (g)$ takes the form of a diagonal matrix. The corresponding infinitesimal generator
is the difference operator $\widehat{K}$ defined by
$$
\left ( \widehat{K}\, \widehat{v} \right ) (n) = -\i n \,\widehat{v}(n)\,.
$$

We call this the {\em realisation in discrete Fourier space}.
As with other realisations, when there is no risk of confusion, we shall drop the hat.
The infinitesimal generators are listed
in Table \ref{infinitesimalGeneratorInDiscreteFourierSpaceTable}.

\begin{table}
\begin{tabular}{c}
\hline
\\ 
$\left ( {K} v \right ) (n) = -\i n\, v (n)$ \\
 \\
$\left ( {A}_1 v \right ) (n) = \frac{1}{2} \left ( \ell +n+1 \right ) v ({n+1})+ \frac{1}{2} \left ( \ell -n+1 \right ) v({n-1})$ \\
\\
$\left ( {A}_2 v \right ) (n) = \frac{\i}{2} \left ( \ell+n+1 \right ) v ({n+1}) - \frac{\i}{2} \left (\ell-n+1 \right )  v({n-1})$ \\
\\
 $\left ( {N}_+ v \right ) (n) = \frac{\i}{2} \left ( \ell +n+1 \right ) v ({n+1})+ \i n \,v(n) - \frac{\i}{2}\left ( \ell-n+1\right ) v({n-1})$ \\
  \\
 $\left ( {N}_- v \right ) (n) = \frac{\i}{2} \left ( \ell + n+1 \right ) v({n+1}) - \i n \,v(n)- \frac{\i}{2} \left ( \ell -n+1 \right ) v({n-1})$ \\
 $ $ \\
 \hline
\\[0.125cm]
\end{tabular}
\caption{Infinitesimal generators associated with the realisation of $T_\ell$ in discrete Fourier space.}
\label{infinitesimalGeneratorInDiscreteFourierSpaceTable}
\end{table}

\begin{proposition}
For $\ell \in {\mathbb N}$, the realisation of $T_\ell^{\pm}$ in discrete Fourier space is unitary
with respect to the inner product
$$
\braket{f | v}_{\pm} := \sum_{\pm n > \ell} \frac{\Gamma(\ell+1 \pm n)}{\Gamma (-\ell \pm n)} \,\overline{f(n)}\, v(n)\,.
$$
\label{discreteHolomorphicProposition}
\end{proposition}

\begin{proof}
The ``infinitesimal method'' reduces the proof to a verification that the infinitesimal generators listed in Table \ref{infinitesimalGeneratorInDiscreteFourierSpaceTable}
are all skew--symmetric with respect to this inner product.
\end{proof}

\section{Growth and fluctuations of products of the type $k \,n_+$}
\label{KNPlusSection}
Let us now use these facts to calculate the generalised
Lyapunov exponent for some probability measures on the group. The measures we shall consider
in this section are supported on the set of matrices of the form
$$
g_j = k(t_j)\,n_+(\tau_j)
$$
where the $\tau_j$ are random variables drawn from an exponential distribution of mean $1/\rho$, and the $t_j$ are drawn independently from some other distribution.
As explained in \S \ref{disorderedSubsection}, the corresponding transfer operator
may then be worked out explicitly in terms of 
$$
{D} := {K}\;\;\text{and}\;\;{E} := {N}_+\,.
$$
We obtain
$$
{\mathscr T}_\ell^{-1} = \left ( 1 -\frac{1}{\rho} {N}_+ \right ) {\mathbb E} \left ( e^{t_j {K}} \right )^{-1}
$$
and so, by making use of the expressions for the operators $K$ and $N_+$ in Table
\ref{infinitesimalGeneratorInDiscreteFourierSpaceTable}, we find that ${\mathscr T}_\ell^{-1}$ is the 
difference operator
\begin{multline}
\left ( {\mathscr T}_\ell^{-1} v \right ) (n) = \frac{v(n)}{\chi(-n)} \\
-\frac{\i}{2 \rho} \left [ \frac{n+1+\ell}{\chi(-n-1)}
v(n+1) + 2 \frac{n}{\chi(-n)} v(n) + \frac{n-1-\ell}{\chi(-n+1)} v(n-1) \right ]
\label{inverseTransferOperatorForK}
\end{multline}
where
\begin{equation}
\chi ( \theta) := {\mathbb E} \left ( e^{\i \theta t_j} \right )
\label{characteristicFunction}
\end{equation}
is the characteristic function of the random variable $t_j$. By using Proposition
\ref{firstAdjointProposition}, we easily work out that the adjoint operator ${\mathscr A}_\ell$
is given by the formula
$$
{\mathscr A}_\ell = {\mathscr A}_0 + \overline{\ell} \,{\mathscr B}
$$
where
\begin{equation}
\left ( {\mathscr A}_0 f \right ) (n) = \frac{1}{\chi(n)} \left \{ f(n) + \frac{\i n}{2 \rho} \left [ f(n+1) + 2 f(n) + f(n-1) \right ] \right \}
\label{A0OperatorForK}
\end{equation}
and
\begin{equation}
\left ( {\mathscr B} f \right ) (n) =  \frac{-\i}{2 \rho} \frac{1}{\chi(n)} \left [ f(n+1) - f(n-1) \right ] \,.
\label{BOperatorForK}
\end{equation}

\subsection{The case $\ell \in {\mathbb N}$}
From the identity (\ref{firstNilpotentIdentity}),
we deduce
\begin{equation}
{N}_\pm \,{{\mathbf e}_{\ell,n}} = \frac{\i}{2} \left [ (\ell+n) \,{{\mathbf e}_{\ell,n-1}} - (\ell-n) \,{{\mathbf e}_{\ell,n+1}}
\pm 2 n \,{{\mathbf e}_{\ell,n}} \right ]\,.
\label{secondNilpotentIdentity}
\end{equation}
We then verify easily that the subspace
$$
V_\ell^0 = \underset{-\ell \le n \le \ell}{\text{span}} \{ {{\mathbf e}_{\ell,n}} \} 
$$
is invariant under ${\mathscr T}_\ell^{-1}$. The corresponding finite-dimensional
representation may be expressed as a tridiagonal matrix of order $2 \ell +1$ with entries
\begin{equation}
\braket{{\mathbf e}_{\ell^\ast, m} | \,{\mathscr T}_\ell^{-1} \,{\mathbf e}_{\ell, n}} 
= a_m  \,\delta_{m,n+1} + b_m \,\delta_{m,n} + c_m \,\delta_{m,n-1} \label{KNtridiagonalMatrix}
\end{equation}
where $\delta_{m,n}$ is the usual Kronecker delta and
\begin{equation}
a_m = \frac{\i}{2 \rho} \,\frac{\ell+1-m}{\chi(-m+1)}\,,\;\;
b_m = \frac{1}{\chi(-m)} \left (1 -\frac{\i \,m}{\rho} \right )\,,\;\;
c_m = \frac{-\i}{2 \rho} \,\frac{\ell+1+m}{\chi(-m-1)}\,.
\label{KNtridiagonalEntries}
\end{equation}

\subsection{Perturbative solution of the adjoint spectral problem}
\label{perturbationSubsectionForK}
We now discuss the practical calculation of the coefficients $\lambda_j$ and $f_j$ in the expansion
(\ref{perturbationExpansion}). If we assume that there is a unique probability
measure on the circle invariant under the action of the matrices 
$k(t_j) \,n_+(\tau_j)$ drawn at random, then the Dyson--Schmidt equation--- which corresponds to the case $j=0$ ---
has one and only one solution that is square-summable; it satisfies
$$
f_0(0)=1 \;\;\text{and}\;\;\lim_{|n| \rightarrow \infty} f_0(n) = 0\,.
$$
Using the fact that $\lambda_0=1$, it is then readily verified by induction on
$j$ that,
for every $j \in \{ 0,\,1,\,\ldots\}$, $\lambda_j$ is a real number and
$$
\forall\, n \in {\mathbb Z}\,,\;\;\overline{f_j(-n)} = f_j(n)\,.
$$
It follows immediately from this symmetry property that,
for a product of type $k \,n_+$ where the $n_+$ component has an exponential
distribution of mean $1/\rho$,
the coefficients in the expansion (\ref{perturbationExpansion}) are given by
\begin{equation}
\lambda_j = \frac{1}{\rho} \text{Im} \left [ f_{j-1}(1) \right ]\,,\;\; j \ge 1\,,
\label{simplifiedEigenvalueFormula}
\end{equation}
where, for $j\ge 1$, $f_j$ is the solution of
\begin{equation}
\left [ \left ( {\mathscr A}_0 - {I} \right ) f_j \right ] (n) = r_j(n) := \sum_{i=1}^j \lambda_i f_{j-i}(n) - \left ( {\mathscr B} f_{j-1} \right ) (n) \,,\;\;n \ge 1\,,
\label{simplifiedEigenvectorFormula}
\end{equation}
subject to
\begin{equation}
f_j(0) = 0 \;\;\text{and}\;\;
\lim_{n \rightarrow \infty} f_j(n) = 0\,.
\label{simplifiedNormalisationCondition}
\end{equation}

The solution of the inhomogeneous problem  (\ref{simplifiedEigenvectorFormula}) is given
by
\begin{equation}
f_j (n) = \sum_{m \ge 1} G(m,n) \,r_j (m)\,,\;\;n \ge 1\,,
\label{formulaForfj}
\end{equation}
where, for a fixed $m \ge 1$, $G(m,\cdot)$ is the solution of
$$
\left [ \left ( {\mathscr A}_0 - {I} \right ) f \right ] (n) = \delta_{m,n}
$$
subject to the conditions
$$
f(0) = 0 \;\;\text{and}\;\;\lim_{n \rightarrow \infty} f(n) = 0\,.
$$

In order to construct this ``Green function'', we require two solutions of the homogeneous
equation
$$
\left [ \left ( {\mathscr A}_0 - {I} \right ) f \right ] (n) = 0\,,\;\; n \ge 1\,.
$$
The particular solution $f_0$ takes care of the condition at infinity. By using ``reduction of order''
and the precise form of the operator ${\mathscr A}_0$ in Equation (\ref{A0OperatorForK}),
we readily find another solution, say $\zeta$, that satisfies the conditions $f(0)=0$
and $f (1) = 1$; it is
given by
\begin{equation}
\zeta (n) = f_0 (n) \sum_{j=1}^n \frac{1}{f_0(j-1) f_0(j)}\,.
\label{zeroSolution}
\end{equation}
The Wronskian of these two solutions is
$$
\begin{vmatrix}
f_0 (n-1) & \zeta(n-1) \\
f_0(n) & \zeta(n)
\end{vmatrix}
= 1
$$
and it follows easily that
\begin{equation}
G (m,n) = 2 \i \rho \,\frac{\chi (m)}{m} \begin{cases}
f_0 (m)\,\zeta (n) & \text{if $n \le m$} \\
f_0 (n) \,\zeta(m) & \text{if $n \ge m$}
\end{cases}\,.
\label{greenFunctionForK}
\end{equation}


In particular, using this in Equation (\ref{formulaForfj}) for $j=1$, we obtain the formula
\begin{multline}
\notag
f_1 (1) = 2 \i \rho \sum_{n=1}^\infty \frac{\chi(n)}{n} f_0 (n) \left \{ \lambda_1 f_0 (n) + \frac{\i}{2 \rho} \frac{1}{\chi(n)}\left [ f_0 (n+1)-f_0(n-1) \right ] \right \} \\
= 2 \i \rho\,\lambda_1 \,\sum_{n=1}^\infty \frac{\chi(n)}{n} f_0^2 (n) - \sum_{n=1}^\infty \frac{1}{n} 
\left [ f_0 (n)\,f_0 (n+1)-f_0(n-1)\,f_0(n) \right ]\,.
\end{multline}
If we use summation by parts for the last series, this becomes
\begin{equation}
f_1 (1) = f_0 (1) + 2 \i \rho\,\lambda_1 \,\sum_{n=1}^\infty \frac{\chi(n)}{n} f_0^2 (n)
- \sum_{n=1}^\infty \frac{1}{n(n+1)} f_0 (n)\,f_0(n+1)
\label{f1ForK}
\end{equation}
and we can, in principle, deduce an expression for $\lambda_2$.

\subsection{Nieuwenhuizen's example}
To give a concrete example, let the random variable $t_j$ have an exponential distribution of mean
density $p$, so that
\begin{equation}
\chi (n) = \frac{1}{1- \frac{\i n}{p}}\,.
\label{exponentialCharacteristicFunction}
\end{equation}
In this case, the Dyson--Schmidt equation (\ref{dysonSchmidtEquation}) becomes, after simplification,
$$
f_0(n+1) + 2 \left ( 1- \frac{\i \rho}{n+\i p} \right ) f_0(n) + f_0(n-1) = 0\,,
$$
subject to $f_0 (0) =1$.
We multiply through by $n+\i p$ and introduce the new unknown
$$
w (n) = \Gamma (n+\i p) f_0(n)\,.
$$
Then
\begin{equation}
\notag
w(n+1) + 2 \left [ n + \i (p-\rho) \right ] w(n) +  (n+\i p) (n +\i p -1) w(n-1) = 0\,.
\end{equation}
The general solution of this difference equation is given explicitly in Masson \cite{Ma} under
the ``Laguerre case'':
\begin{multline}
\notag
w (n) = c_1 \,(-1)^n  \,\Gamma (n+\i p) \,\Gamma (n+1+\i p)\,W_{-n -\i p, \frac{1}{2}} \left ( - 2 \i \rho \right  ) \\
+ c_2  \,(-1)^n\,\Gamma (n+1+\i p)\, M_{-n -\i p, \frac{1}{2}} \left ( - 2 \i \rho \right  )\,.
\end{multline}
The eigenvector $f_0$ is the solution of the Dyson--Schmidt equation that is recessive as $n \rightarrow \infty$; it is obtained by taking $c_2=0$. Hence
\begin{equation}
f_0(n) = 
(-1)^n  \frac{\Gamma (n+\i p+1) \,W_{-n -\i p, \frac{1}{2}} \left ( - 2 \i \rho \right  )}{\Gamma (\i p+1) \,W_{-\i p, \frac{1}{2}} \left ( - 2 \i \rho \right )}\,,\;\;n \ge 0\,.
\label{dysonSchmidtSolutionForK}
\end{equation}
Equation (\ref{simplifiedEigenvalueFormula}) then becomes
$$
\lambda_1 = \frac{1}{\rho} \,\text{Im} \left [ (-1) (\i p+1) \frac{W_{-\i p-1,\frac{1}{2}}(-2 \i \rho)}{W_{-\i p,\frac{1}{2}}(-2 \i \rho)} \right ]\,.
$$
Now, Formula 3 in \cite{GR}, \S 9.234, says that
$$
z \,W_{\kappa,\mu}'(z) = \left ( \kappa - \frac{z}{2} \right ) W_{\kappa,\mu}(z)
-\left [ \mu^2 - \left ( \kappa-\frac{1}{2} \right )^2 \right ]\,W_{\kappa-1,\mu} (z)\,.
$$
We may therefore express this result
in the alternative form
\begin{equation}
\lambda_1 = \frac{2}{p} \,\text{Im} \left [ \frac{W_{-\i p,\frac{1}{2}}'(-2 \i \rho)}{W_{-\i p,\frac{1}{2}}(-2 \i \rho)} \right ]
\label{firstNieuwenhuizenFormula}
\end{equation}
and we recover  a formula that was first obtained by Nieuwenhuizen in \cite{Ni}, \S 5.2,
by a different method. We shall return to this example in \S \ref{nieuwenhuizenRevisitedSubsection} and obtain an
alternative expression for $\lambda_2$ in terms of an integral.

\section{Realisation on the line}
\label{lineSection}
Another useful realisation of $T_\ell$ is that associated with the line of equation 
$$
x_2 =1\,.
$$
Every function $v \in V_\ell$ is completely determined
by the values  it takes on this line: Indeed, for $x_2 >0$,
$$
v(x_1,x_2) = x_2^{2 \ell} \,v \left ( \frac{x_1}{x_2}, 1 \right )
$$
whilst, for $x_2 < 0$, we have instead
$$
v(x_1,x_2) = v(-x_1,-x_2)= (-x_2)^{2 \ell} \,v \left ( \frac{x_1}{x_2}, 1 \right )\,.
$$
Hence, for every $x_2 \ne 0$,
$$
v(x_1,x_2) = |x_2|^{2 \ell} \,\widetilde{v} \left ( \frac{x_1}{x_2} \right )
$$
where
$$
\widetilde{v}(x) := v(x,1)\,.
$$
This defines $\widetilde{v}$ for every $x \in {\mathbb R}$. We denote by $Q$ the map that assigns $\widetilde{v}$ to $v$ and define the space $\widetilde{V}_\ell$ of complex-valued functions on ${\mathbb R}$ by Equation (\ref{realisationSpace}). As an example, we have
$$
\left [ Q \,1_{\ell} \right ] (x) = \widetilde{1}_\ell (x) := \left (1+x^2 \right)^\ell\,.
$$

\begin{lemma}
For $\widetilde{v} :\,{\mathbb R} \rightarrow {\mathbb C}$ to belong to $\widetilde{V}_\ell$, it is necessary and sufficient that
both $\widetilde{v}$ and the function
$$
x \mapsto |x|^{2\ell} \,\widetilde{v} (1/x)
$$
be smooth.
\label{representationSpaceLemma}
\end{lemma}
One consequence of this characterisation is that, for a function $\widetilde{v}$ to belong to $\widetilde{V}_\ell$, it is necessary that
the limits
\begin{equation}
\lim_{x \rightarrow \infty} |x|^{-2 \ell} \,\widetilde{v}(x) \;\;\text{and}\;\;\lim_{x \rightarrow -\infty} |x|^{-2 \ell} \,\widetilde{v}(x)
\label{characterisationAtInfinity}
\end{equation}
exist and be equal.

The map $Q :\, V_{\ell} \rightarrow \widetilde{V}_{\ell}$ is bijective and the
representation $T_{\ell}$ is realised on the space $\widetilde{V}_{\ell}$ via Equation (\ref{intertwinning}) or, more explicitly,
\begin{equation}
\left [ \widetilde{T}_\ell(g) \widetilde{v} \right ] (x) := \left [ T_\ell(g) v \right ] \left ( x, 1\right ) 
= v \left ( a x + c, \,b x + d  \right )
= \left | b x + d\right |^{2 \ell}\,
 \widetilde{v}(x \cdot g)
\label{lineRepresentation}
\end{equation}
where
\begin{equation}
x \cdot g := \frac{a x + c}{b x+d}
\label{actionOnTheLine}
\end{equation}
describes the action of the group on the line. We shall henceforth drop the tilde. The infinitesimal generators associated with this realisation are
listed in Table \ref{infinitesimalGeneratorTable}. In particular, every eigenvector of $J_0 = \i \,K$
corresponding to the eigenvalue $n \in {\mathbb Z}$ is proportional to
\begin{equation}
{{\mathbf e}_{\ell,n}}(x)  = \frac{1}{\sqrt{\pi}} \left ( \frac{x+\i}{x-\i} \right )^n 1_\ell (x) = \frac{1}{\sqrt{\pi}} (x+\i)^{\ell+n} ( x-\i)^{\ell-n}
\label{basisFunctionOnTheLine}
\end{equation}
and the ${{\mathbf e}_{\ell,n}}$ form a basis for $V_\ell$.

\subsection{Multiplier, dual space and adjoint}
The Lebesgue measure $dx$ is invariant
under the action of the subgroup $N_-$. The
corresponding multiplier is
\begin{equation}
\sigma_{N_-} (x,g) := \frac{d}{d x} \left (  x \cdot g \right ) = \frac{1}{\left ( b x +d \right )^2}
\label{LineMultiplier}
\end{equation}
and so the representation may also be expressed as
\begin{equation}
\left [ T_\ell(g) v \right ] (x) = \sigma_{N_-} (x,g)^{-\ell} \,v ( x \cdot g )\,.
\label{LineRealisationWithMultiplier}
\end{equation}

We now turn to the determination of the space dual to $V_\ell$. Since
$$
-2 \ell - 2 \overline{\ell^\ast} = 2\,,
$$
we deduce from Equation (\ref{characterisationAtInfinity}) that, for every $v \in V_{\ell}$ and every $f \in V_{\ell^\ast}$, the limits
$$
\lim_{x \rightarrow \infty} |x|^{2} \,\overline{f(x)}\,v(x) 
\;\;\text{and}\;\;\lim_{x \rightarrow -\infty} |x|^{2} \,\overline{f(x)}\,v(x) 
$$
exist and are equal. It follows that the 
function $\overline{f}\,v$ is absolutely integrable, and so $f$ may be identified with the 
continuous linear functional defined on $V_\ell$ by
\begin{equation}
v \mapsto  \int_{-\infty}^\infty \overline{f(x)}\,v(x)\,dx =: \braket{f|v}\,.
\label{lineInnerProduct}
\end{equation}
Hence
$$
V_{\ell^\ast} \subset V_\ell^\ast\,.
$$
By adapting for this new definition of $\braket{\cdot | \cdot}$
the calculation we performed earlier for the realisation on the circle, we readily
deduce the analogue of Equation (\ref{adjointOnTheCircle}) for this realisation:
\begin{equation}
\left [ T_\ell^\ast (g) f \right ] (x) = 
\sigma_{N_-} (x,g^{-1})^{-\ell^\ast} \,f (x \cdot g^{-1}) = \left ( a - bx  \right )^{2 \ell^\ast} \,f \left ( \frac{d x-c}{a-bx} \right )\,.
\label{adjointOnTheLine}
\end{equation}
.

The foregoing suggests that the ``natural'' Hilbert space for the realisation on the line is
$$
H := \left \{ v :\,{\mathbb R} \rightarrow {\mathbb C} :\, \int_{-\infty}^\infty |v(x)|^2\, dx < \infty \right \}
$$
equipped with the inner product (\ref{lineInnerProduct}).
We remark, however, that, in this realisation, $V_\ell \subset H$ only if $\text{Re} \,\ell < -1/2$. Nevertheless,
the sequences
$$
\{ {\mathbf e}_{\ell^\ast,m} \}_{m \in {\mathbb Z}} \subset V_{\ell^\ast}\;\;\text{and} \;\;\{ {{\mathbf e}_{\ell,n}} \}_{n \in {\mathbb Z}}
\subset V_{\ell}
$$
constitute a bi-orthogonal system for this inner product since, in this realisation,
\begin{multline}
\label{biorthogonality}
\braket{{\mathbf e}_{\ell^\ast,m} | \,{\mathbf e}_{\ell, n}} := \frac{1}{\pi} \,\int_{-\infty}^\infty dx\,\overline{\left ( \frac{x+\i}{x-\i} \right )^m1_{\ell^\ast}(x)}
\,\left ( \frac{x+\i}{x-\i} \right )^n1_{\ell}(x) \\
=  \frac{1}{\pi}\,\int_{-\infty}^\infty \frac{dx}{1+x^2} \overline{\left ( \frac{x+\i}{x-\i} \right )^m}\, \left ( \frac{x+\i}{x-\i} \right )^n \overset{\underset{\downarrow}{x = \cot \frac{\theta}{2}}}{=} 
\frac{1}{2 \pi} \int_0^{2 \pi} e^{\i (n-m) \theta} \,d \theta = \delta_{m,n}\,.
\end{multline}

\subsection{Realisation in continuous Fourier space}
\label{continuousFourierSubsection}
The realisation on the line is convenient if we wish to diagonalise the representation with respect to the 
subgroup $N_-$. Indeed, for
$$
g = n_-(t) := \begin{pmatrix}
1 & 0 \\
t & 1
\end{pmatrix}
$$
we have
$$
\left [ T_{\ell} \left ( g\right ) v \right ] (x) = v(x+t)\,.
$$
The restriction of $T_\ell$ to $N_-$ acts on $V_{\ell}$ by translation, and this suggests that we go over to the 
Fourier transform:
\begin{equation}
\widehat{v}(s) := \int_{-\infty}^\infty d x \,e^{\i s x} v(x)\,,\;\;s \in {\mathbb R}\,.
\label{fourierTransform}
\end{equation}
Proceeding formally, if we denote by $Q$ the operator that assigns to $v$ its Fourier transform 
$\widehat{v}$ and put
\begin{equation}
\widehat{V}_{\ell} := \left \{ \widehat{v} :\, v \in V_\ell \right \}
\label{transformSpace}
\end{equation}
we obtain a
new realisation $\widehat{T}_\ell$ of the representation, defined by
$$
\widehat{T}_\ell\, Q = Q\, T_\ell\,.
$$ 
In particular, for $g=n_-(t)$,
$$
\left [ \widehat{T}_\ell \left ( g \right ) \widehat{v} \right ] (s) = \left [ Q \,T_\ell \left ( g\right ) v \right ] (s) = \int_{-\infty}^\infty dx \,e^{\i s x} v(x+t) 
\overset{\underset{\downarrow}{x' = x+t}}{=}
e^{-\i s t} \,\widehat{v} (s) 
$$
and so $\widehat{T}_\ell(g)$ is indeed a multiplication operator.

Now, the fact that the limits (\ref{characterisationAtInfinity}) exist and are equal imply in particular that
$\widehat{v}$ {\em is typically not smooth at the origin}. 
For $\text{Re} \,\ell \ge -1/2$,  
$v$ does not even decay sufficiently
quickly at infinity to have a Fourier transform in the classical sense.
Nevertheless, we can give a precise meaning to Formula (\ref{fourierTransform})
by viewing the elements of $V_\ell$ as ``generalised functions'' in the sense of Gel'fand \& Shilov \cite{GS}.
Put briefly, a generalised function, say $v$, is a continuous linear functional
on the space of smooth functions $\varphi$ with compact support, and one expresses the value
that the functional takes at $\varphi$ in the form
$$
\int_{-\infty}^\infty d x \,\overline{v(x)}\,\varphi(x) \,.
$$
The derivative $v'$ of the generalised function $v$ is then defined to be the generalised function 
$$
\varphi \mapsto -\int_{-\infty}^\infty d x\,\overline{v(x)}\, \varphi' (x)\,.
$$

We remark that the Fourier transform $\widehat{\varphi}$ of a test function $\varphi$ exists in the classical sense.
The Fourier transform of the generalised function $v$ is then, by definition, a continuous linear
functional on the space of transformed test functions, given implicitly 
by the Parseval formula
\begin{equation}
\int_{-\infty}^\infty d s \, \overline{\widehat{v}(s)} \,\widehat{\varphi}(s)  = 2 \pi \int_{-\infty}^\infty d x \, \overline{v(x)}\,\varphi(x) \,
\label{parsevalFormula}
\end{equation}
for all transformed test functions $\widehat{\varphi}$. For ordinary, square-integrable functions, this definition coincides with the classical
definition. Furthermore, the Fourier transform of the generalised function $v'$ is $-\i s \widehat{v}$, and the Fourier transform of the generalised function $x v$ is $-\i \widehat{v}'$,
where $\widehat{v}$ is the Fourier transform of the generalised function $v$.

With these clarifications, the Fourier transform 
$\widehat{\mathbf e}_{\ell,n}$ of the basis function ${\mathbf e}_{\ell, n}$
can be worked out as follows: from the formula (12) in \cite{Er}, \S 3.2, we obtain
\begin{equation}
\widehat{\mathbf e}_{\ell,n} (s) = \sqrt{\pi} (-1)^n \left | s/2 \right |^{-\ell-1}
\begin{cases}
\frac{1}{\Gamma(n-\ell)}\,W_{n,\ell+\frac{1}{2}} \left ( 2 s \right ) & \text{if $s > 0$} \\
\frac{1}{\Gamma(-n-\ell)}\,W_{-n,\ell+\frac{1}{2}} \left ( -2 s \right ) & \text{if $s<0$}
\end{cases}\,.
\label{fourierTransformBasis}
\end{equation}
This formula is valid in the classical sense if $\text{Re} \,\ell < -1/2$ and $\ell$ is not an integer.
If $\ell$ is an integer, the formula remains valid if we interpret $\Gamma(m)$ as infinity for 
$m=-1,\,-2\,,\ldots$

As a concrete illustration, when $\ell$ is a negative integer, the Fourier transform exists in the classical sense. Using the formula
$$
W_{\frac{\alpha}{2}+\frac{1}{2}+m,-\frac{\alpha}{2}} (z) = (-1)^m m! \,e^{-z/2} z^{\frac{\alpha+1}{2}}
L_m^{\alpha} (z)\,,\;\;m=0,\,1,\,2,\,\ldots
$$
where
$$
L_m^{\alpha} (z) := \sum_{j=0}^m (-1)^j \binom{j+\alpha}{m-j} \frac{z^j}{j!}\,,
$$
we can express some of the Fourier coefficients in terms of Laguerre polynomials:
\begin{equation}
\widehat{\mathbf e}_{\ell,n}(s) = 2 \sqrt{\pi}\,(-1)^\ell \frac{\Gamma(n+\ell+1)}{\Gamma(n-\ell)}\,s^{-2\ell-1} \,e^{-s}\,L_{n+\ell}^{-2 \ell-1} ( 2 s)\,\theta(s)\,,\;\; n \ge -\ell, 
\label{firstFourierTransform}
\end{equation}
where
$$
\theta(s) := \begin{cases}
1 & \text{if $s >0$} \\
0 & \text{otherwise}
\end{cases}
$$
is the Heaviside function.
For $\ell < n <-\ell$, the formula (\ref{fourierTransformBasis}) may be used verbatim, whilst
$$
\widehat{\mathbf e}_{\ell,n}(s) = \sqrt{\pi} \frac{(-1)^n}{\Gamma (-n-\ell)} \left | \frac{s}{2} \right |^{-\ell-1}  \,W_{-n,\ell+\frac{1}{2}} (-2 s) \, \theta(-s) \;\;\text{for $n \le \ell$}\,.
$$

When $\ell=0$, the Fourier transform must be interpreted in the sense of generalised functions.
To compute it, we can use the fact that
$$
{\mathbf e}_{0,n}(x) = (1+x^2) \,{\mathbf e}_{-1,n}(x) \implies \widehat{\mathbf e}_{0,n} (s)= \left ( 1 - \frac{d^2}{d s^2} \right ) \widehat{\mathbf e}_{-1,n}(s)\,.
$$
In particular, this yields
$$
\widehat{\mathbf e}_{0,0}(s) = \frac{1}{\sqrt{\pi}}\,\delta (s)
$$
where $\delta$ is the familiar Dirac delta.

In summary, by going over to the Fourier transform, we obtain from the realisation on the line a new realisation
of the representation $T_\ell$ which has the property of being diagonal with respect to the subgroup
$N_-$. We call it the {\em realisation in continuous Fourier space}, and we shall henceforth drop the hats.
In this realisation, the sequences
$$
\left \{ {\mathbf e}_{\ell^\ast,m} \right \}_{m \in {\mathbb Z}}\;\;\text{and}\;\;\left \{ {{{\mathbf e}_{\ell,n}}}\right \}_{n \in {\mathbb Z}}
$$
defined by Formula (\ref{fourierTransformBasis}) constitute, by virtue of Parseval's identity, a bi-orthogonal system for the inner product
$$
\braket{{f} | {v}} := \frac{1}{2 \pi}\,\int_{-\infty}^{\infty} d s \,\overline{f(s)}\,{v}(s)\,.
$$
The infinitesimal generators associated 
with this realisation are listed in Table \ref{infinitesimalGeneratorInContinuousFourierSpaceTable}.

\begin{table}
\begin{tabular}{c}
 \hline
\\ 
${K} = \i  \left [ \frac{s}{2} \frac{d^2}{d s^2} + (\ell+1) \frac{d}{d s} - \frac{s}{2} \right ]$ \\
 \\
${A}_1= - s \frac{d}{ds} - (\ell+1)$ \\
\\
${A}_2  = -\i \left [ \frac{s}{2} \frac{d^2}{d s^2} + (\ell+1) \frac{d}{ds} + \frac{s}{2} \right ]$ \\
\\
 ${N}_+ = -\i \left [ s \frac{d^2}{ds^2} + 2 (\ell+1) \frac{d}{ds}  \right ]$ \\
  \\
 ${N}_- = -\i s$ \\
 $ $ \\
 \hline
\\[0.125cm]
\end{tabular}
\caption{Infinitesimal generators associated with the realisation of $T_\ell$ in continuous Fourier space.}
\label{infinitesimalGeneratorInContinuousFourierSpaceTable}
\end{table}

\subsection{Realisation of one of the holomorphic representations in continuous Fourier space}
Let now
$$
\ell \in \left \{ -1,\,-2,\,\ldots \right \}\,.
$$
In this case, the invariant subspace
$$
V_\ell^- := \underset{n \ge -\ell}{\text{span}} \{ {{\mathbf e}_{\ell,n}} \} 
$$
consists of the Fourier transform of those functions on the real line that may be continued to the lower half of the complex plane; it
yields one of the holomorphic subrepresentations, denoted $T_\ell^-$.
The basis functions were worked out in \S \ref{continuousFourierSubsection}; we found
that they are supported on ${\mathbb R}_+$ and given by Formula (\ref{firstFourierTransform}).
By using the fact that the Laguerre polynomials $L_m^\alpha$ are orthogonal with respect to the inner
product
$$
\left ( \varphi, \psi \right ) := \int_0^\infty \varphi(x) \,\psi(x)\,x^{\alpha} e^{-x} dx
$$
we deduce that the
basis $\{ {{\mathbf e}_{\ell,n}} \}_{n \ge -\ell}$ is orthogonal with respect to the inner
product
\begin{equation}
\braket{f | v}_\ell := \int_0^\infty \overline{f(s)} \,v(s)\, s^{2 \ell +1} ds
\label{continuousFourierSpaceInnerProduct}
\end{equation}
Now, the infinitesimal generators are skew-symmetric with respect to this inner product; it follows that the realisation of $T_\ell^-$ in continuous Fourier space is unitary.
The connection between the holomorphic representations of $\text{SL}(2,{\mathbb R})$ and the Laguerre polynomials
has been discussed by Davidson {\em et al.} \cite{DOZ}.

\section{Growth and fluctuations of products involving $N_-$}
\label{NminusSection}

In this section, we study products consisting of matrices of the form
$$
g_j = n_- (t_j)\,e(\tau_j)
$$
where $e(\tau)$ is a fixed one-parameter subgroup, the $\tau_j$ are random variables drawn from an exponential distribution of mean $1/\rho$, and the $t_j$ are drawn independently from some other distribution with characteristic function 
$\chi$. We shall consider in turn the cases
$e(\tau_j)=k(\tau_j)$, $n_+(\tau_j)$ and $a_1(\tau_j)$.

\subsection{Products of the type $n_- \,k$}
In this case, the transfer operator is given by
$$
{\mathscr T}_\ell = {\mathbb E} \left ( e^{t_j {N_-}} \right ) \left ( 1 -\frac{1}{\rho} {K} \right )^{-1}\,. 
$$

Let us first discuss the case $\ell \in {\mathbb N}$.
The restriction of ${\mathscr T_\ell}$ to the invariant subspace $V_\ell^0$ may then be expressed as a product of two matrices of order $2 \ell +1$. To the $k$ component, there corresponds the diagonal matrix
$$
\text{diag} \left ( \frac{1}{1+\i \ell/\rho},\, \cdots,\, \frac{1}{1-\i \ell/\rho}\right )\,.
$$
To the $n_-$ component, there corresponds the matrix with entries
\begin{equation}
\braket{\ell^\ast, m | \,{\mathbb E} \left ( e^{t_j N_-} \right ) | \ell, n} 
= {\mathbb E} \left [ \braket{\ell^\ast, m | e^{t_j N_-} | \ell, n} \right ]
= \sum_{i=0}^{2 \ell} \frac{{\mathbb E} (t^i)}{i!} N_-^i
\notag
\end{equation}
where we have used the fact that
\begin{equation}
N_- := \frac{\i}{2} \begin{pmatrix} 2 \ell & 1 & 0 & & & \\
-2 \ell & 2 (\ell-1) & 2 & & & \\
 & \ddots & \ddots & \ddots & & \\
& & & -2 & -2(\ell-1) & 2 \ell \\
& & & 0 & -1 & -2 \ell
\end{pmatrix}
\label{NtridiagonalEntries}
\end{equation}
is a nilpotent matrix such that $N_-^i = 0$ for $i > 2 \ell$.
It follows in particular that, for $\ell \in {\mathbb N}$, the generalised Lyapunov exponent
depends only on the first $2 \ell$ moments of the random variable $t_j$.

Turning next to the perturbative solution of the adjoint spectral problem, we see,
with the help of Table \ref{infinitesimalGeneratorInContinuousFourierSpaceTable}, that the adjoint ${\mathscr A}_\ell$ of ${\mathscr T}_\ell^{-1}$
is given by the formula
$$
{\mathscr A}_\ell = {\mathscr A}_0 + \overline{\ell} \,{\mathscr B}
$$
where
\begin{equation}
{\mathscr A}_0  = \frac{1}{\chi(s)} \left [ 1 + \frac{\i s}{2 \rho} \left ( \frac{d^2}{d s^2} - 1 \right )\right ]\,,\;\; {\mathscr B} = -\frac{1}{\chi(s)} \frac{\i}{\rho} \frac{d}{d s}
\label{operatorsForNminusK}
\end{equation}
and
$$
\chi ( \theta) := {\mathbb E} \left ( e^{\i \theta t_j} \right )
$$
is the characteristic function of the random variable $t_j$. 

We can look for a perturbative solution of the spectral problem of the form
(\ref{perturbationExpansion}) 
and, as before, this leads to the recurrence relation
(\ref{recurrenceRelation}) for the unknown coefficients $\lambda_j$ and $f_j$. For $j=0$, in view of the normalisation (\ref{perturbativeNormalisationCondition}),
we then have that
$$
\lambda_0 = 1\,,\;\;v_0 = \delta\,,
$$
whilst $f_0$ solves the Dyson--Schmidt equation (\ref{dysonSchmidtEquation})
subject to the conditions $f_0(0) =1$ and $f_0(s) \rightarrow 0$ as $|s| \rightarrow \infty$.

For $j=1$, however, a new situation arises, namely
that {\em the formula (\ref{eigenvalueFormula}) does not make sense in this realisation} since, as the expressions for the basis functions make clear,
$f_{0}$ has a jump in its derivative, so that the Dirac delta functional cannot be applied to it. This difficulty is inherent in the use of the projective variable $x$ associated with the realisation on the line; it is linked to the fact that, in contrast with the realisation 
on the circle, the representation space now depends explicitly on $\ell$, so that it is not immediately clear in what space one should seek the successive terms $f_j$ of the perturbation expansion. We proceed to describe the adjustments that are needed in order to compute them. 

The problem of assigning a precise meaning to the formula (\ref{eigenvalueFormula}) has already been addressed by 
Halperin \cite{Ha} in connection with the Frisch--Lloyd model. Indeed, for the realisation on the line, the integral 
$$
\int_{-\infty}^\infty x f(x)\, dx
$$
does not always exist in the strict sense for $f \in V_{-1}$, but the fact that the limits
$$
\lim_{x \rightarrow -\infty} x^2 f(x) \;\;\text{and}\;\;\lim_{x \rightarrow \infty} x^2 f(x)
$$
exist and are equal implies that it always makes sense as
a Cauchy principal value integral. For the realisation in continuous Fourier space, the corresponding interpretation is to view $\braket{{\mathscr B} f_{0} | v_0}$ as
$$
\overline{\frac{1}{2} \left [ {\mathscr B} f_{0} (0+)
+ {\mathscr B} f_{0} (0-) \right ]} = \frac{\i}{2 \rho} \overline{\left [f_{j-1}'(0+) + f_{j-1}'(0-) \right ]}\,.
$$
Now, it is easy to show that the solution of the Dyson--Schmidt equation has the property
$$
\forall\, s \in {\mathbb R}\,,\;\;\overline{f_0(-s)} = f_0 (s)\,.
$$
For $j=1$, we may therefore use the equivalent interpretation
\begin{equation}
\braket{{\mathscr B} f_{0} | v_0} := \frac{1}{\rho} \,\text{Im} \left [ f_{0}'(0+) \right ]\,.
\label{regularisationFormula}
\end{equation}
Equation (\ref{eigenvalueFormula}) then says
\begin{equation}
\lambda_1 = \frac{1}{\rho} \,\text{Im} \left [ f_{0}'(0+) \right ]\,.
\label{simplifiedEigenvalueFormulaForNminus}
\end{equation}

Proceeding to the next stage, it follows from the fact that $\lambda_1$ is real that
\begin{equation}
r_1 := \lambda_1 f_0 - {\mathscr B} f_0 = \lambda_1 f_0 + \frac{1}{\chi(s)} \frac{\i}{\rho} f_0' 
\label{formulaForr1}
\end{equation}
has the property
$$
\forall\, s \in {\mathbb R}\,,\;\;\overline{r_1(-s)} = r_1(s)\,.
$$
The solution $f_1$ of Equation (\ref{recurrenceRelation}) inherits this property, and the equation may therefore be replaced by
\begin{equation}
\left ( {\mathscr A}_0 - {I} \right ) f_1  = r_1\,,\;\;s > 0\,,
\label{simplifiedEigenvectorFormulaForNminus}
\end{equation}
subject to $f_1(0)=0$ and $f_1(s) \rightarrow 0$ as $s \rightarrow \infty$.
The solution of this inhomogeneous problem  is given by
\begin{equation}
f_1 (s) = \int_0^\infty dt \,G(s,t) \,r_1(t)\,,\;\;s > 0\,,
\label{formulaForf1}
\end{equation}
where, for a fixed $t > 0$, $G(\cdot,t)$ is the solution of
$$
\left [ \left ( {\mathscr A}_0 - {I} \right ) f \right ] (s) = \delta (s-t)
$$
subject to the conditions
$$
f(0) = 0 \;\;\text{and}\;\;\lim_{s \rightarrow \infty} f(s) = 0\,.
$$

The construction of this Green function follows essentially the same principles
as in \S \ref{perturbationSubsectionForK}: From the solution $f_0$ of the Dyson--Schmidt equation, we 
obtain a second solution
\begin{equation}
\zeta (s) = f_0 (s) \int_{0}^s \frac{d \tau}{f_0^2( \tau)}\,.
\label{zeroSolutionForNminus}
\end{equation}
of the homogeneous problem such that $\zeta(0)=0$ and $\zeta'(0+) =1$.
The Wronskian of these two solutions is
$$
\begin{vmatrix}
f_0 (s) & \zeta(s) \\
f_0'(s) & \zeta'(s)
\end{vmatrix}
= 1
$$
and it follows easily that
\begin{equation}
G (s,t) = 2 \i \rho \,\frac{\chi (t)}{t} \begin{cases}
f_0 (t)\,\zeta (s) & \text{if $s \le t$} \\
f_0 (s) \,\zeta(t) & \text{if $s \ge t$}
\end{cases}\,.
\label{greenFunctionForNminus}
\end{equation}

Equation (\ref{formulaForf1}) then yields, after differentiation with respect to $s$, 
\begin{equation}
\notag
f_1' (s) = 2 \i \rho \left \{ f_0'(s) \int_0^s \frac{dt}{t} \,\chi(t) \,\zeta(t) \,r_1 (t) + \zeta'(s) 
\int_s^\infty \frac{dt}{t} \,\chi(t) \,f_0 (t) \,r_1 (t) \right \}\,,\;\;s > 0\,.
\end{equation}
The right-hand side is a complex-valued fonction of $s$; whether it has a limit as
$s \rightarrow 0+$ depends on the behaviour of $r_1$ near $0$. Now, Formula (\ref{simplifiedEigenvalueFormulaForNminus})
implies that the {\em real part} of $r_{1}$ vanishes at $0$.
We deduce that the {\em imaginary part} of the expression for $f_1'$
does have a limit, and we may write
$$
\lim_{s \rightarrow 0+}  \text{Im} \left [ f_1'(s) \right ] =  2 \rho \int_0^\infty \frac{d s}{s}\,\text{Re} \left [ \chi(s)\, f_0 (s)\,r_1(s) \right ]\,.
$$
Upon using the regularisation (\ref{regularisationFormula}) again, we find
\begin{equation}
\lambda_2  = \frac{1}{\rho} \,\lim_{s \rightarrow 0+}  \text{Im} \left [ f_{1}'(s) \right ] 
= 2 \int_0^\infty \frac{ds}{s} \,\text{Re} \left [ \lambda_1 \chi(s) f_0^2(s) + \frac{\i}{\rho} f_0(s) f_0'(s) \right ]\,.
\label{lambda2ForNminus}
\end{equation}
This procedure may be continued to compute successively $f_2$, $\lambda_3$ and so on.

\subsection{Nieuwenhuizen's example revisited}
\label{nieuwenhuizenRevisitedSubsection}
We remark that
$$
g_j = n_- (t_j)\,k(\tau_j) \implies g_j^t = k(-\tau_j) \,n_+(t_j)
$$
and, since taking the transpose does not change the generalised Lyapunov
exponent, we can use Nieuwenhuizen's example to check that the realisations
in discrete and in continuous Fourier spaces yield consistent results.

Let the random variable $t_j$ have an exponential distribution of mean
density $p$, so that
$$
\chi (s) = \frac{1}{1- \frac{\i s}{p}}\,.
$$
The Dyson--Schmidt equation (\ref{dysonSchmidtEquation}) then takes the form
$$
f_0'' - \left ( \frac{2 \rho}{p-\i s} + 1 \right ) f_0 = 0\,,
$$
subject to $f_0 (0) =1$ and $f_0(s) \rightarrow 0$ as $s \rightarrow 0$. This equation may be reduced to Whittaker's equation, and we easily deduce
\begin{equation}
f_0(s) = \frac{W_{-\i \rho, \frac{1}{2}} \left ( 2 \i p + 2 s \right  )}{W_{-\i \rho, \frac{1}{2}} \left (2 \i p \right )}\,,\;\;s > 0\,.
\label{dysonSchmidtSolutionForNminus}
\end{equation}
Equation (\ref{simplifiedEigenvalueFormulaForNminus}) then yields
\begin{equation}
\lambda_1 = \frac{2}{\rho} \,\text{Im} \left [ \frac{W_{-\i \rho, \frac{1}{2}}' \left ( 2 \i p \right  )}{W_{-\i \rho, \frac{1}{2}} \left (2 \i p \right )}\right ]
\label{secondNieuwenhuizenFormula}
\end{equation}
and this is indeed consistent with the formula (\ref{firstNieuwenhuizenFormula}) we found earlier. 

Plots of $\gamma$ and of the variance $\sigma^2 = \lambda_1^2/4- \lambda_2/2$ against $1/\rho$, for a fixed $p$, are
shown in Figure \ref{frischLloydFigure}.

\begin{figure}
\includegraphics[width=8cm,height=6cm]{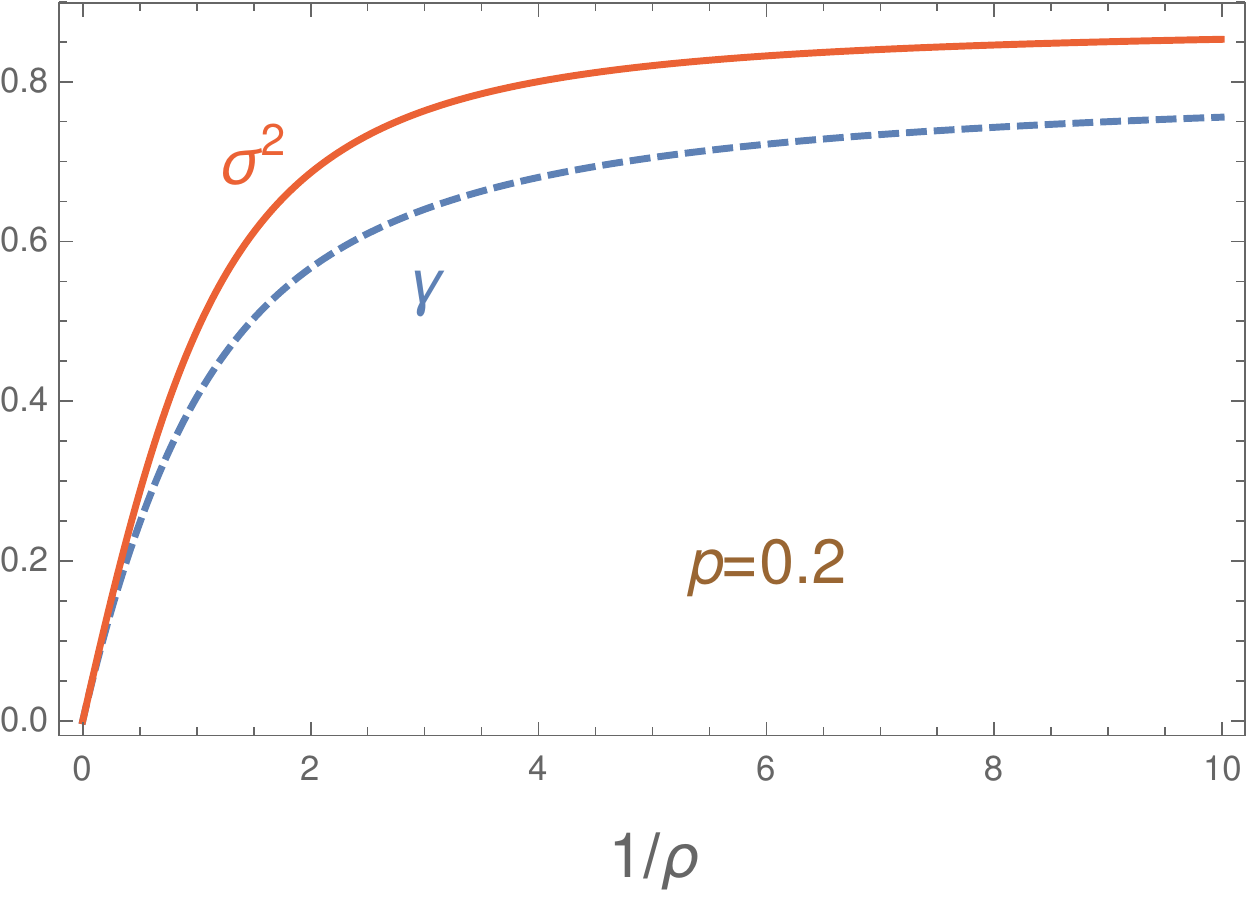}%
%
\caption{The growth rate $\gamma$ and the variance $\sigma^2$ for the product of type
$n_-(t_j) \,k (\tau_j)$ discussed in \S \ref{nieuwenhuizenRevisitedSubsection}, plotted against the mean value $1/\rho$ of $\tau_j$, for a fixed unit mean $1/p=5$ of $t_j$.}
\label{frischLloydFigure}
\end{figure}

%
%

\subsection{Products of the type $n_- \,n_+$}
\label{dysonSubsection}
The perturbative analysis of products consisting of elements of the form
$$
g_j = n_- (t_j)\,n_+(\tau_j)\,,
$$
where $t_j$ and $\tau_j$ are as before, requires no further adjustment other than
updating the definitions of the operators ${\mathscr A}_0$ and ${\mathscr B}$:
\begin{equation}
{\mathscr A}_0  = \frac{1}{\chi(s)} \left ( 1 - \frac{\i s}{\rho} \frac{d^2}{d s^2} \right )\,,\;\; {\mathscr B} = \frac{1}{\chi(s)} \frac{2 \i}{\rho} \frac{d}{d s}\,.
\label{operatorsForNminusNplus}
\end{equation}
This leads to the revised formulae
\begin{equation}
\lambda_1 = -\frac{2}{\rho} \,\text{Im} \left [ f_0'(0+)\right ]
\label{dysonFormulaForLambda1}
\end{equation}
and
\begin{equation}
\lambda_2 =  -\frac{2}{\rho} \,\lim_{s \rightarrow 0+}  \text{Im} \left [ f_{1}'(s) \right ] = 2 \int_0^\infty \frac{ds}{s} \,\text{Re} \left [ \lambda_1 \chi(s) f_0^2(s) - \frac{2 \i}{\rho} f_0(s) f_0'(s) \right ]\,.
\label{dysonFormulaForLambda2}
\end{equation}

As an illustration, we can solve Exercise 5.5 in \cite{BL}, which corresponds to taking for $t_j$ an exponential distribution of mean $1/p$. The solution of the
Dyson--Schmidt equation is
\begin{equation}
f_0 (s) = \frac{\sqrt{p-\i s} \,K_1 \left ( 2 \sqrt{\rho \,(p-\i s)} \right )}{\sqrt{p}\,K_1 \left ( 2 \sqrt{\rho \,p} \right )}\,.
\label{letacSeshadriExample}
\end{equation}
It follows in particular that the invariant distribution is a generalised inverse Gaussian
--- a fact discovered by Letac \& Seshadri \cite{LS}.
After some straightforward manipulations, we deduce from Equations (\ref{growthRateAndVarianceInTermsOfLambda}) and  (\ref{dysonFormulaForLambda1})
$$
\gamma = -\frac{\lambda_1}{2} = \frac{1}{\sqrt{\rho\,p}} \,\frac{K_0 \left (2 \sqrt{\rho\,p} \right )}{K_1 \left (2 \sqrt{\rho\,p} \right )}\,.
$$
Plots of $\gamma$ and of the variance $\sigma^2 = \lambda_1^2/4- \lambda_2/2$ against $1/\rho$, for a fixed $p$, are
shown in Figure \ref{dysonFigure}.

\begin{figure}
\includegraphics[width=8cm,height=6cm]{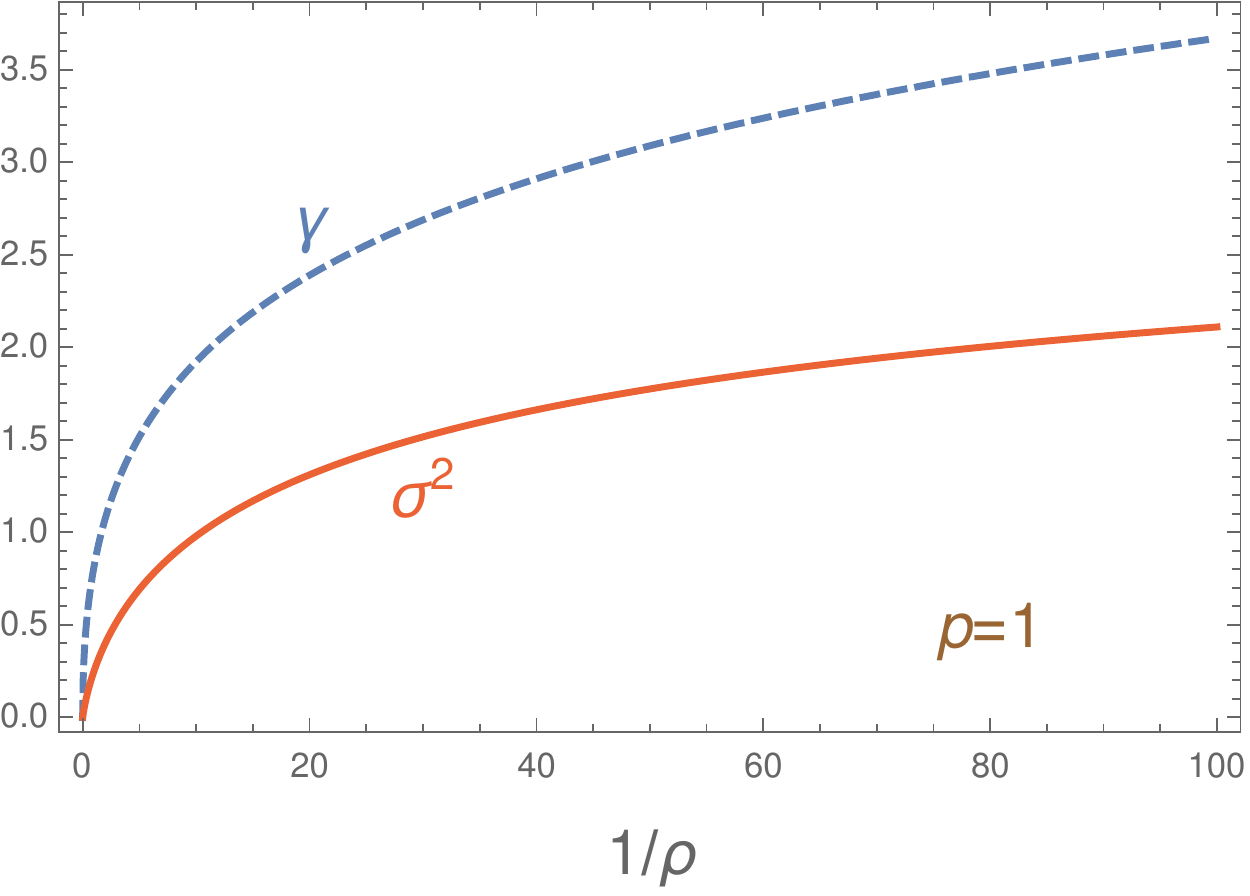}%
%
\caption{The growth rate $\gamma$ and the variance $\sigma^2$ for the product of type
$n_-(t_j) \,n_+ (\tau_j)$ discussed in \S \ref{dysonSubsection}, plotted against the mean value $1/\rho$ of $\tau_j$, for a fixed unit mean $1/p=1$ of $t_j$.}
\label{dysonFigure}
\end{figure}

If we replace $\rho$ by $-\rho$, we recover the model studied by Kotani in \cite{Ko},
\S 5, Example 1, for which the Lyapunov exponent may be expressed in terms of a
Hankel function.

\subsection{Products of the type $n_-  \,a_1$} 
The matrices are
lower triangular in this case; the diagonal entries depend only on the $a_1$ component, and so the calculation of the generalised Lyapunov exponent is trivial. For instance, if
$t_j$ is exponentially
distributed with mean $1/p$, we deduce directly from the definition (\ref{generalisedLyapunovExponent}) that
$$
\Lambda (2 \ell) = -\ln \left ( 1- \frac{\ell}{\rho} \right )\;\;\text{for $\text{Re} \,\ell < \rho$}\,.
$$
This simplification manifests itself in the fact that, as can be seen from the entries of
Table \ref{infinitesimalGeneratorInContinuousFourierSpaceTable}, the inverse of the transfer operator
$$
{\mathscr T}_\ell^{-1} 
=\left ( 1 -\frac{1}{\rho} {A}_1 \right )   \left ( 1 -\frac{1}{p} {N}_- \right )  
$$
and hence also its adjoint ${\mathscr A}_\ell$ are both {\em first-order differential operators}:
$$
{\mathscr T}_\ell^{-1} = \left [ 1 + \frac{1}{\rho} \left ( s \frac{d}{d s} + \ell+1 \right ) \right ] \left (
1+\frac{\i s}{p} \right )\,,
\;\;
{\mathscr A}_\ell = \left (
1-\frac{\i s}{p} \right ) \left [ 1 - \frac{1}{\rho} \left ( s \frac{d}{d s} - \overline{\ell} \right ) \right ] \,.
$$

Despite such simplicity, it turns out that the perturbative approach 
does not produce the correct result in this case.
The Dyson--Schmidt equation
$$
\frac{d f_0}{d s} + \frac{\rho}{s+\i p} f_0 = 0
$$
is readily solved by separating the variables:
$$
f_0(s) = c (s + \i p)^\rho
$$
for some constant $c$. But this function cannot be the Fourier transform of a probability
density.
We are in a situation where {\em there is no measure on the Furstenberg boundary invariant for this product.} The probability distribution on $\text{SL}(2,{\mathbb R})$
corresponding to this random product fails to satisfy Furstenberg's ``strong irreduciblity
condition''. This condition says that the smallest subgroup containing the distribution's support, acting on ${\mathbb R}^2$, should not leave
invariant a finite union of one-dimensional subspaces \cite{Fu}; without it, and without some other conditions, the existence of an invariant measure cannot be guaranteed. We refer the reader to \cite{BL}, Theorem 4.1 and Proposition 4.3, for a useful list of such conditions.

If one changes the probability distribution
so that it is now $-\tau_j$ that is exponentially distributed with mean $1/\rho$ then
the generalised Lyapunov exponent does not change but
the invariant measure does exist. Indeed, this change is tantamount to replacing 
$\rho$ by $-\rho$ in the Dyson--Schmidt equation; the solution is now given by
$$
f_0 (s) = \frac{c}{(s + \i p)^\rho}
$$
and, after normalisation, we recognise the Fourier transform of a gamma density.
We can use the recurrence relation (\ref{recurrenceRelation}) with
$$
v_0 = \delta\;\;\text{and}\;\;{\mathscr B} = - \frac{1}{\rho} {I}
$$
and the perturbative approach produces the correct result
$$
\lambda_1 = -\frac{1}{\rho} \;\;\text{and}\;\;\lambda_2 = 0\,.
$$

\section{Realisation on the hyperbola}
\label{hyperbolaSection}
For the diagonalisation of $T_\ell$ with respect to the subgroup $A_1$, it is convenient to work with the realisation on the two-branched hyperbola of equation 
$$
x_2  = \frac{1}{|x_1|}\,.
$$
To every $v \in V_\ell$, we associate the function 
$\widetilde{v} :\,{\mathbb R}_\ast \rightarrow {\mathbb C}$ defined by
$$
\widetilde{v}(x) := v \left ( \frac{x}{\sqrt{|x|}}, \frac{1}{\sqrt{|x|}} \right )\,.
$$
We denote by $Q$ the operator that assigns $\widetilde{v}$ to $v$ and 
define the space $\widetilde{V}_\ell$ of complex-valued functions on ${\mathbb R}_\ast$ by Equation (\ref{realisationSpace}). 
By making use of the identity
$$
\widetilde{v} (x) = |x|^{-\ell} \,v (x,1)
$$
we easily deduce the following from Lemma \ref{representationSpaceLemma}:
\begin{lemma}
For $\widetilde{v}$ to belong to $\widetilde{V}_\ell$, it is necessary and sufficient that both the functions
$$
{\mathbb R}_\ast \ni x \mapsto |x|^\ell \,\widetilde{v}(x)\;\;\text{and}\;\;
{\mathbb R}_\ast \ni x \mapsto |x|^\ell \,\widetilde{v}(1/x)
$$
be smooth.
\label{representationSpaceLemmaForTheHyperbola}
\end{lemma}
In particular, if $\widetilde{v} \in \widetilde{V}_\ell$, then the limits
$$
\lim_{x \rightarrow \infty} |x|^{-\ell}\,\widetilde{v}(x)
\;\;\text{and}\;\;
\lim_{x \rightarrow -\infty} |x|^{-\ell}\,\widetilde{v}(x)
$$
exist and are equal. Furthermore, the limit
$$
\lim_{x \rightarrow 0} |x|^\ell\,\widetilde{v}(x)
$$
also exists.

The map $Q :\, V_{\ell} \rightarrow \widetilde{V}_{\ell}$ is bijective and the
representation $T_{\ell}$ is realised on the space $\widetilde{V}_{\ell}$ via Equation (\ref{intertwinning}):
\begin{equation}
\left [ \widetilde{T}_\ell(g) \widetilde{v} \right ] (x) :=  
 \left | \frac{(a x+ c)(b x + d)}{x} \right |^{\ell}\,
 \widetilde{v}\left ( \frac{a x + c}{b x+d} \right )\,.
\label{hyperbolaRepresentation}
\end{equation}
From now on, we shall omit the tilde.

We remark that the measure $dx/|x|$ is invariant under the action
of the subgroup $A_1$. The corresponding multiplier is given by the formula
$$
\sigma_{A_1} (x,g) := \frac{|x|}{\left | x \cdot g \right |} \frac{d}{d x} \left (  x \cdot g \right ) = \left |\frac{x}{(a x+c) ( b x +d)} \right |
$$
and so the realisation on the hyperbola may also be expressed as
\begin{equation}
\left [ T_\ell(g) v \right ] (x) = \sigma_{A_1} (x,g)^{-\ell} \,v ( x \cdot g )\,.
\label{hyperbolaRealisationWithMultiplier}
\end{equation}

\subsection{The adjoint}
For every $v \in V_{\ell}$ and every $f \in V_{\ell^\ast}$, we deduce from the asymptotic
behaviours found above that
$$
|x| \,\overline{f(x)}\,v(x) \xrightarrow[|x| \rightarrow \infty]{} C
\;\;\text{and}\;\;\frac{1}{|x|} \overline{f(x)}\,v(x) \xrightarrow[x \rightarrow 0]{} c
$$
for some constants $c$ and $C$ depending on $f$ and $v$. It follows
that $f$ may be identified with the 
continuous linear functional defined on $V_\ell$ by
\begin{equation}
v \mapsto \int_{{\mathbb R_\ast}} \overline{f(x)}\,v(x)\,\frac{dx}{|x|} =: \braket{f|v}\,.
\label{hyperbolaInnerProduct}
\end{equation}
Hence
$$
V_{\ell^\ast} \subset V_\ell^\ast\,.
$$
It is then readily verified by direct calculation that the adjoint of $T_\ell$ with respect
to this inner product is given by
\begin{multline}
\left [ T_\ell^\ast (g) f \right ] (x) = \sigma_{A_1} ( x,g^{-1})^{-\ell^\ast} f (x \cdot g^{-1} ) \\
= \left | \frac{(dx-c)(a - bx)}{x}  \right |^{\ell^\ast} \,f \left ( \frac{d x-c}{a-bx} \right )\,.
\label{adjointOnTheHyperbola}
\end{multline}

The appropriate Hilbert space for the realisation on the hyperbola is thus
$$
H := \left \{ v :\,{\mathbb R}_\ast \rightarrow {\mathbb C} :\, \int_{{\mathbb R}_\ast} |v(x)|^2\, \frac{dx}{|x|} < \infty \right \}
$$
equipped with the inner product (\ref{hyperbolaInnerProduct}).
In this realisation, $V_\ell \subset H$ only if $\text{Re} \,\ell < 0$.
The sequences
$$
\{ {\mathbf e}_{\ell^\ast,m}\}_{m \in {\mathbb Z}} \subset V_{\ell^\ast}\;\;\text{and} \;\;\{ {{\mathbf e}_{\ell,n}} \}_{n \in {\mathbb Z}}
\subset V_{\ell}
$$
constitute a bi-orthogonal system for this inner product. 

\subsection{Realisation in Mellin space}
In order to diagonalise with respect to the subgroup $A_1$, we work with the realisation on the hyperbola.
First, we assign, to every $v \in V_\ell$, the pair
$$
(v_+,\,v_-) : \, (0,\infty) \rightarrow {\mathbb C}^2
$$
where
$$
v_+(x) := v(x)\;\;\text{and}\;\;v_-(x) := v(-x)\,,\;\;x>0\,.
$$
We can assert that there exist constants
$c$ and $C$, depending only on $v$, such that
$$
v_{\pm} (x) \sim C \,x^\ell \;\;\text{as $x \rightarrow \infty$} \;\; \text{and}\;\;
 v_{\pm} (x) \sim c \,x^{-\ell} \;\;\text{as $x \rightarrow 0+$}\,.
$$
It then follows that the Mellin transforms
$$
\widehat{v}_\pm (s) := \int_0^\infty dx\, x^{\i s-1} v_{\pm}(x)
$$
exist, in the classical sense, in the strip
\begin{equation}
\text{Re} \,\ell < \text{Im} \,s < -\text{Re} \,\ell\,.
\label{mellinStrip}
\end{equation}
We put 
$$
\widehat{v} := (\widehat{v}_+,\widehat{v}_-)
$$
and denote by $Q$ the operator that assigns $\widehat{v}$ to $v$. We may then obtain a new realisation $\widehat{T}_\ell$ of
$T_\ell$ in the space
$$
\widehat{V}_\ell := \left \{ Q v :\, v \in V_\ell \right \}
$$
by setting
$$
\widehat{T}_\ell \,Q = Q \,T_\ell\,.
$$
We call this the realisation in {\em Mellin space}. 
In particular, for
$$
g = \begin{pmatrix}
e^{\frac{t}{2}} & 0 \\
0 & e^{-\frac{t}{2}}
\end{pmatrix}\,,
$$
we have
$$
\left [ T_\ell (g) v \right ]_{\pm} (x) = v_{\pm} \left ( e^t x \right )\;\;\text{for}\; x >0\,. 
$$
We readily deduce
\begin{equation}
\notag
\left [ \widehat{T}_\ell (g) \widehat{v} \right ] (s) = 
e^{-\i s t} \,\widehat{v}(s)
\end{equation}
and so this realisation is indeed diagonal with respect to the subgroup $A_1$.

We have defined the Mellin transform $\widehat{v}_{\pm}$ in such a way that it is, for $s$ real, the Fourier transform of the function
${\mathbb R} \ni t \mapsto v_{\pm}(e^t)$. Parseval's formula for the Fourier transform then leads to
the identity
\begin{equation}
\braket{f | v} := \int_{-\infty}^\infty \frac{d x}{|x|} \,\overline{f(x)} \,v(x) 
= \frac{1}{2 \pi} \int_{-\infty}^\infty d s \left [ \overline{\widehat{f}_+(s)}\,\widehat{v}_+(s) +
\overline{\widehat{f}_-(s)}\,\widehat{v}_-(s)  \right ]
\label{mellinParsevalFormula}
\end{equation}
and we shall use the right-hand side as the definition of the inner product in $\widehat{V}_\ell$.
This choice ensures the bi-orthogonality of $\{ \widehat{\mathbf e}_{\ell^\ast,m} \}_{m \in {\mathbb Z}}$
and $\{ \widehat{\mathbf e}_{\ell,n} \}_{n \in {\mathbb Z}}$.

It is interesting to examine the form that the ${{\mathbf e}_{\ell,n}}$ assume when we go over to the realisation in Mellin
space: By definition,
$$
{\mathbf e}_{\ell,n {\pm}}(x) = \frac{(-1)^n}{\sqrt{\pi}} \left ( \frac{1 \mp \i x}{1 \pm \i x} \right )^n \left ( x+\frac{1}{x} \right )^\ell\,.
$$
Hence
$$
\widehat{\mathbf e}_{\ell,n \pm}(s) =
\frac{(-1)^n}{\sqrt{\pi}} \int_0^\infty dx \,x^{\i s - \ell-1} (1 \mp \i x)^{\ell+n} (1 \pm \i x)^{\ell-n}\,.
$$
The value of this integral is given by Formula 6.2.35 in \cite{Er}:
\begin{equation}
\widehat{\mathbf e}_{\ell,n \pm}(s) = \frac{(-1)^n}{\sqrt{\pi}} \,e^{\pm \i \frac{\pi}{2}(\ell-\i s)} \,{\tt B} \left ( -\ell+\i s,\,-\ell-\i s \right )
{_2} F_1 \left (-\ell - n, \i s-\ell; -2 \ell; 2  \right )\,.
\label{hypergeometricFormula}
\end{equation}
For $\ell \in \{-1,\,-2,\,\ldots \}$ and $n \ge -\ell$,
the hypergeometric function in this formula may be expressed in terms of the Meixner--Pollaczek
polynomials \cite{Koo}. An alternative inner product may then be found that makes $\widehat{T}_\ell^-$ unitary.

It would be tedious to write down $\widehat{T}_\ell(g)$ explicitly for an arbitrary $g \in G$; see
\cite{Vi,VK1}.
For our purposes, we only need to know the infinitesimal generators associated with the various subgroups, and these can be worked out by taking the Mellin transform of the corresponding
generator for $T_\ell$. Thus, for instance, in order to compute $\widehat{K}$, we use
the fact that, for $x>0$,
$$
\left ( {K} v \right )_{\pm} (x) =  \left ( {K} v \right ) (\pm x) = \pm \left [ \frac{\ell}{2} \left ( \frac{1}{x} - x \right ) v_\pm(x) + \frac{1+x^2}{2}
\,v_\pm'(x) \right ]\,.
$$
If we then multiply each of these identities by $x^{\i s-1}$ and integrate over $x$, we obtain
$$
\left ( \widehat{K} \,\widehat{v} \right )_{\pm}(s) = \pm \frac{1}{2}( \ell+1-\i s) \,\widehat{v}_{\pm}(s+\i) \mp
\frac{1}{2} (\ell+1+\i s) \,\widehat{v}_{\pm}(s-\i)\,.
$$
The other generators are listed in Table \ref{infinitesimalGeneratorInMellinSpace} where, as is our practice when there is no risk of confusion, we have dropped the hats.

\begin{table}
\begin{tabular}{c}
\hline
\\ 
$\left ( {K} \,{v} \right )_{\pm}(s) = \pm \frac{1}{2}( \ell+1-\i s) \,{v}_{\pm}(s+\i) \mp
\frac{1}{2} (\ell+1+\i s) \,{v}_{\pm}(s-\i)$ \\
 \\
$\left ( A_1 {v} \right )_{\pm} = -\i s\,v_{\pm}(s)$ \\
\\
$\left ( A_2 \,{v} \right )_{\pm}(s) = \pm \frac{1}{2}( \ell+1-\i s) \,{v}_{\pm}(s+\i) \pm
\frac{1}{2} (\ell+1+\i s) \,{v}_{\pm}(s-\i)$ \\
\\
 $\left ( N_+ {v} \right )_{\pm}(s) = \pm \left ( \ell +1 + \i s \right ) v_{\pm}(s-\i)$ \\
  \\
 $\left ( N_- {v} \right )_\pm (s)= \pm \left ( \ell+1 -\i s \right ) v_{\pm}(s+\i)$ \\
 $ $ \\
  \hline
\\[0.125cm]
\end{tabular}
\caption{Infinitesimal generators associated with the realisation of $T_\ell$ in Mellin space.}
\label{infinitesimalGeneratorInMellinSpace}
\end{table}

Although in principle this realisation should be advantageous for the study of products 
involving the subgroup $A_1$, the difference operators that arise have the awkward feature
that the increments are imaginary, whereas the auxiliary condition for the Dyson--Schmidt
equation is that the solution should decay as the independent variable $s$ tends to infinity
along the real axis. This makes it difficult to identify the correct solution without resorting to
the realisation on the hyperbola; see \cite{CTT} for examples 
of products of the type $a_1 \,a_2$.
This difficulty is exacerbated when we try to compute the next terms in the perturbation
expansion.

\section{Concluding remarks}
\label{conclusionSection}
In this paper, we have shown how the calculation of the Lyapunov exponent for certain probability distributions on the group
$\text{SL}(2,{\mathbb R})$ can be carried out successfully by using a family of representations in a space of functions on the Furstenberg boundary. 
The cases that we studied are of interest as models of one-dimensional
systems and led to simple formulae for the growth rate and the variance;
we found that the special functions that appear in these formulae are the same as those
that serve to express the basis functions in a suitable realisation of the representation space. One can reasonably expect
that the calculation of the generalised Lyapunov exponent can be pushed further. 
Admittedly, the probability distributions considered here
had very special features; in every case, essential use was made of the univariate exponential distribution, and it is not clear why this distribution should play a distinguished part in connection with $\text{SL}(2,{\mathbb R})$. 

Furstenberg's theory applies to a large class of semi-simple groups. For example, the Furstenberg boundary of the real symplectic group--- which provides a generalisation
of $\text{SL}(2,{\mathbb R})$ that is of particular interest in the context of
disordered systems--- is the (isotropic) flag manifold \cite{CL}. The so-called
``degenerate principal series representations'' are defined in spaces of functions
on (parts of) that manifold and bear some resemblance with the family of representations considered here; they have been studied by Lee \cite{Le}, following
ideas put forward by Howe \& Tan \cite{HT}.  Likewise, Vilenkin's ideas
are not restricted to the group $\text{SL}(2,{\mathbb R})$
and have been developed extensively in recent years \cite{VK1,VK2}.
One can therefore envisage the existence of special functions associated with other semi-simple groups that might be useful in the study of random
products.  
It is apparent, however, that any extension of our results to larger groups would require a fairly detailed knowledge
of their representation theory, as well as a high level of computational skill.

\bibliographystyle{amsplain}

\end{document}